\documentclass[final]{siamltex}
\usepackage{epsfig}
\usepackage{graphicx}
\usepackage{amsmath}
\usepackage{amssymb}
\usepackage{setspace}
\usepackage{enumerate}
\usepackage{array}
\usepackage{color}
\newcolumntype{C}[1]{>{\centering\let\newline\\\arraybackslash\hspace{0pt}}m{#1}}
\usepackage[ruled,linesnumbered,vlined]{algorithm2e}

\newtheorem{defn}{\noindent $\mathbf{Definition}$}[section]

\newtheorem{prop}[defn]{$\mathbf{Proposition}$}
\newtheorem{thm}[defn]{$\mathbf{Theorem}$}

\title{Fast Spherical Quasiconformal Parameterization of Genus-0 Closed Surfaces with Application to Adaptive Remeshing}

\author{Gary Pui-Tung Choi\thanks{John A. Paulson School of Engineering and Applied Sciences, Harvard University (pchoi@g.harvard.edu).} 
\and Mandy Hiu-Ying Man\thanks{Department of Mathematics, The Chinese University of Hong Kong (hyman@link.cuhk.edu.hk).} 
\and Lok Ming Lui\thanks{Department of Mathematics, The Chinese University of Hong Kong (lmlui@math.cuhk.edu.hk).}
}

\begin{document}
\maketitle
\begin{abstract}
In this work, we are concerned with the spherical quasiconformal parameterization of genus-0 closed surfaces. Given a genus-0 closed triangulated surface and an arbitrary user-defined quasiconformal distortion, we propose a fast algorithm for computing a spherical parameterization of the surface that satisfies the prescribed distortion. The proposed algorithm can be effectively applied to adaptive surface remeshing for improving the visualization in computer graphics and animations. Experimental results are presented to illustrate the effectiveness of our algorithm. 

\end{abstract}


\pagestyle{myheadings}
\thispagestyle{plain}
\markboth{Choi, Man and Lui}{Fast Spherical Quasiconformal Parameterization}

\section{Introduction}
In recent decades, conformal parameterization of genus-0 closed meshes has been widely studied by various research groups. Also, various quasiconformal parameterization algorithms have been developed for planar domains and simply-connected open meshes by different researchers. However, the study of quasiconformal parameterization on meshes with spherical topology is limited. Given a user-defined quasiconformal distortion, we aim to compute a spherical quasiconformal parameterization with the prescribed distortion. In this work, we first develop the concept of quasiconformal dilation on triangulated meshes as a measurement of quasiconformality. Then, we propose a fast algorithm for the computation of spherical parameterizations that satisfy arbitrary user-defined quasiconformal dilations. In particular, a uniform quasiconformal dilation results in a spherical parameterization with uniform conformality distortion. %To the best of the authors' knowledge, this is the first work on the computation of spherical quasiconformal parameterizations of genus-0 closed surfaces.

With our proposed spherical quasiconformal parameterization algorithm, adaptive surface remeshing can be easily achieved. In computer graphics and animations, the visual quality of surfaces is affected by the piecewise linear discretization of them. Different triangulations (or quadrangulations) of a surface can have significantly different visual effects. In general, regular triangles are preferred as they can provide a smoother approximation of the original surfaces. However, this may not be true in some special cases. For instance, sharp triangles may be more suitable for approximating a sharp and narrow feature on a surface. To adaptively produce different types of triangles on different parts of a surface, we can apply our proposed spherical quasiconformal parameterization algorithm with certain user-defined distortions. Then, using standard remeshing techniques such as the spherical Delaunay triangulation scheme, we can obtain a triangulation on the spherical parameterization. This induces a triangulation on the original surface. Since the parameterization is quasiconformal but not necessarily conformal, the induced triangulation may not be Delaunay. Instead, if the user-defined distortion is assigned in a special way, the induced triangulation will accomplish our goal.

The organization of the paper is as follows. In Section \ref{contribution}, we highlight the contribution of our work. In Section \ref{previous}, we review the literature related to our work. The mathematical background of our work is introduced in Section \ref{background}. In Section \ref{main}, we explain the details of our proposed algorithm with application to remeshing. Experimental results are presented in Section \ref{experiment}. In Section \ref{conclusion}, we conclude the paper and outline the future work.

\section{Contributions} \label{contribution}
Our proposed spherical quasiconformal parameterization algorithm has following advantages:
\begin{enumerate}
 \item \emph{Efficiency}: Our proposed algorithm only involves solving a few sparse linear systems and hence is highly efficient in practice. 
 \item \emph{Bijectivity}: The bijectivity of the resulting parameterization is supported by quasiconformal theory.
 \item \emph{Accuracy}: Our algorithm can accurately compute a spherical parameterization with the prescribed distortion.
 \item \emph{Applicability}: Our algorithm can be effectively applied for adaptively remeshing genus-0 closed surfaces.
\end{enumerate}

\section{Previous works}\label{previous}
\subsection{Conformal and quasiconformal parameterization} 
\begin{table}[t]
    \centering
    \begin{tabular}{ |C{30mm}|C{25mm}|C{20mm}|C{30mm}|C{23mm}| }
    \hline
    Method & Topology & Parameter domain & Distortion criterion & Iterative minimization required? \\ \hline
    Extremal Quasiconformal Maps \cite{Weber12}  & Disk-type & Plane & Uniform Conformality Distortion & Yes \\ \hline
    Bounded Distortion Mappings \cite{Lipman12}  & Disk-type & Plane & Quasiconformal & Yes \\ \hline
    Discrete Curvature Flow \cite{Zeng12} & Disk-type & Plane & Quasiconformal & Yes \\ \hline
    Injective and Bounded Distortion Mappings \cite{Aigerman13} & Disk-type / Genus-0 & Plane / Polycube & Quasiconformal & Yes \\ \hline
    QC Iteration \cite{Lui14} & Disk-type & Plane & Uniform Conformality Distortion & Yes \\ \hline
    TEMPO \cite{Meng15} & Disk-type & Plane & Uniform Conformality Distortion & Yes \\ \hline
    Our proposed FSQC algorithm & Genus-0 & Sphere & Quasiconformal / Uniform Conformality Distortion & No \\ \hline
    \end{tabular}
    \bigbreak
    \caption{Several works on quasiconformal parameterization of simply-connected surfaces.}
    \label{previouswork}
\end{table}

In the past two decades, surface conformal parameterization has been widely studied \cite{Floater02,Floater05,Sheffer06,Hormann07}. In particular, the recent approaches of conformal parameterizations include simplifying harmonic energy minimization \cite{Lai14,Choi15d}, generalizing Ricci flow to the discrete setting \cite{Jin08,Yang09,Zhang14}, and introducing quasiconformal composition \cite{Choi15a,Choi15b,Choi15c}.

In recent years, the study of surface quasiconformal parameterizations has been emerging. The works on quasiconformal parameterization of simply-connected surfaces are summarized in Table \ref{previouswork}. In \cite{Weber12}, Weber et al. introduced an algorithm for computing extremal quasiconformal mappings for simply-connected open meshes using holomorphic quadratic differentials. In \cite{Lipman12}, Lipman introduced bounded distortion mappings for triangular meshes with boundary. Zeng et al. \cite{Zeng12} proposed to compute quasiconformal parameterizations using a discrete auxiliary metric and the Yamabe flow. In \cite{Aigerman13}, Aigerman and Lipman developed an algorithm for computing bounded distortion mappings in 3D. The algorithm can be applied for parameterizing meshes onto the 2D plane or polycubes. In \cite{Lui14}, Lui et al. proposed an iterative algorithm for computing Teichm\"uller maps, which are with uniform conformality distortion, of simply-connected open meshes. The convergence of the algorithm has been proved in \cite{Lui15}. In \cite{Meng15}, Meng et al. proposed the TEMPO algorithm for computing landmark-matching Teichm\"uller parameterization of disk-type point cloud surfaces.

\subsection{Remeshing via parameterization} 
%The quality of triangulation on surface is important for the accuracy of the numerical computation and analysis. 
Surface remeshing has been widely studied for generating desired surface meshes in recent decades. In particular, surface remeshing is usually achieved with the aid of parameterization. For instance, Hormann et al. \cite{Hormann01} studied the remeshing for topologically disk-like surfaces with a boundary and no holes using parameterization over a planar domain. They applied the Most Isometric Parameterization Strategy (MIPS) \cite{Hormann00} for generating triangle meshes with subdivision connectivity. Gu et al. \cite{Gu02} proposed to remesh a surface onto a completely regular structure called geometry image, by cutting the mesh along a network of edge paths into a topological disk and computing a square parameterization. In \cite{Praun03}, Praun et al. introduced the idea of remeshing genus-0 closed surfaces by spherical parameterization instead of planar parameterization. This avoids cutting the surface and hence the parameterization becomes unconstrained. Hu et al. \cite{Hu11} proposed a low-distortion spherical parameterization for closed genus 0 meshes to generate subdivision connectivity meshes. The meshes are then smoothed by the umbrella operator. Remacle et al. \cite{Remacle10} developed a scheme based on one-to-one discrete harmonic maps for generating surface meshes. In \cite{Choi14}, Choi et al. proposed an algorithm to compute planar conformal parametrization of disk-type meshes and to obtain regular triangulations on the planar domain using landmark-matching Teichm\"uller maps.

\section{Mathematical background} \label{background}

In this section, we introduce the concept of conformal maps and quasiconformal maps. Readers are referred to \cite{Schoen94,Schoen97,Gardiner00,Jost11,Lui13} for more details.

\subsection{Conformal maps}
We begin with the definition of conformal maps between Riemann surfaces.
\begin{defn}[Conformal maps]
Let $\mathcal{M}$ and $\mathcal{N}$ be two Riemann surfaces. A map $f:\mathcal{M} \to \mathcal{N}$ is \emph{conformal} if there exists a scalar function $\lambda(x^1,x^2)>0$, called the {\it conformal factor}, such that
\begin{equation}
f^*ds_{\mathcal{N}}^2 = \lambda ds_{\mathcal{M}}^2.
\end{equation}
\end{defn}
An immediate consequence is that every conformal map preserves angles and hence the infinitesimal shapes of the surface.

Among all conformal maps, we are particularly interested in those which map an arbitrary genus-0 closed surface onto a simple standard domain. The existence of such conformal maps is guaranteed by the uniformization theorem.

\begin{thm}[Uniformization of Riemann surfaces]
\label{thm:uniformization}
Every simply connected Riemann surface $\mathcal{M}$ is conformally equivalent to exactly one of the following three domains:
\begin{enumerate}[(i)]
 \item the Riemann sphere,
 \item the complex plane,
 \item the open unit disk.
\end{enumerate}
\end{thm}
As our focus in this work is genus-0 closed surfaces, it is natural to consider the unit sphere as a standard parameter domain. Now, the problem is how to find a spherical conformal map.
This can be done by considering harmonic maps.

\begin{defn}[Harmonic maps]
 The {\it Dirichlet energy} for a map $f: {\mathcal{M}} \to \mathcal{N}$ is defined as
\begin{equation}\label{eqt:harmonic}
E(f) = \int_{\mathcal{M}} ||\nabla f||^2 dv_{\mathcal{M}}.
\end{equation}
In the space of mappings, the critical points of $E(f)$ are called {\it harmonic maps}.
\end{defn}

On triangulated meshes, the discrete Dirichlet energy is given by
\begin{equation}
E(f) = \sum_{[u,v] \in K} k_{uv} ||f(u)-f(v)||^2.
\end{equation}
Here $k_{uv} = \cot \alpha + \cot \beta$, where $\alpha,\beta$ are the angles opposite to the edge $[u,v]$. 

Consequently, the discretization of the Laplacian is given by
\begin{equation}
\Delta f = \sum_{[u,v] \in K} k_{uv} (f(u)-f(v)).
\end{equation}

For genus-0 closed surfaces, conformal maps are equivalent to harmonic maps \cite{Jost11}. Hence, the problem of finding a conformal map between two genus-0 closed surfaces is equivalent to an energy minimization problem.

\subsection{Quasiconformal maps}
In this section, we introduce the concept of quasiconformal maps, a generalization of conformal maps, and the related properties.

\begin{defn}[Quasiconformal maps]
A map $f: \mathbb{C} \to \mathbb{C}$ is said to be \emph{quasiconformal(QC)} if it satisfies the Beltrami equation
\begin{equation}\label{eqt:beltrami}
\frac{\partial f}{\partial \overline{z}} = \mu(z) \frac{\partial f}{\partial z}
\end{equation}
for some complex-valued function $\mu$ satisfying $||\mu||_{\infty}< 1$, and $\frac{\partial f}{\partial z}$ is non-vanishing almost everywhere. Here, the complex partial derivatives are defined by
\begin{equation}
 \frac{\partial f}{\partial z} := \frac{1}{2} \left(\frac{\partial f}{\partial x}  - i \frac{\partial f}{\partial y} \right)
 \ \ \text{ and } \ \
 \frac{\partial f}{\partial \overline{z}} := \frac{1}{2} \left(\frac{\partial f}{\partial x}  + i \frac{\partial f}{\partial y} \right).
\end{equation}
\end{defn}

$\mu$ is called the \emph{Beltrami coefficient} of the quasiconformal map $f$. $f$ is conformal around a small neighborhood of $p$ if and only if $\mu(p) = 0$, as Equation (\ref{eqt:beltrami}) becomes the Cauchy-Riemann equation in this situation. Hence, the Beltrami coefficient $\mu$ is closely related to the conformality distortion of $f$.

Besides, Beltrami coefficients are also related to the bijectivity of their associated quasiconformal maps, as explained by the following theorem.
\begin{thm}\label{thm:bijectivity}
If $f:\mathbb{C} \to \mathbb{C}$ is a $C^1$ map satisfying $\|\mu_f\|_{\infty} <1$, then $f$ is bijective.
\end{thm}

In addition, the \emph{maximal quasiconformal dilation} of $f$ is given by
\begin{equation} \label{eqt:maximalqcdilation}
K = \frac{1+||\mu||_{\infty}}{1-||\mu||_{\infty}}.
\end{equation}
A geometrical illustration of quasiconformal maps is shown in Figure \ref{fig:qcmap}.

\begin{figure}[!t]
 \centering
 \includegraphics[width=0.8\textwidth]{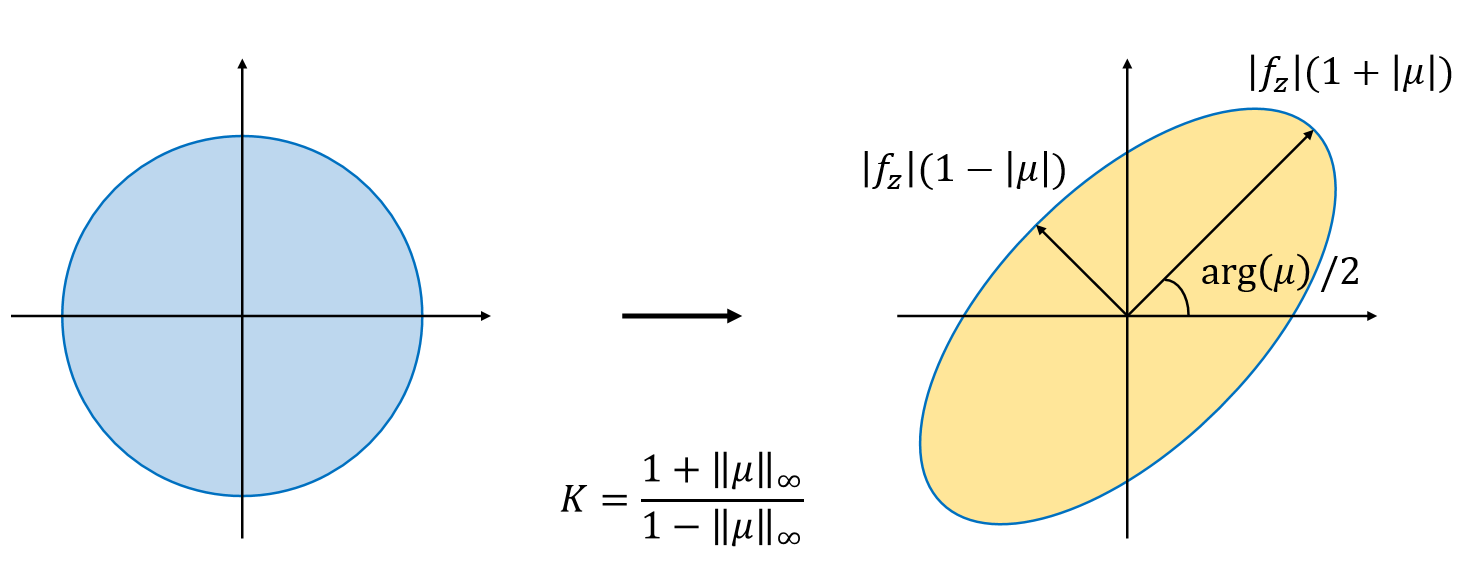}
 \caption{An illustration of quasiconformal maps.}
 \label{fig:qcmap}
\end{figure}

% Now, we consider the composition of two quasiconformal maps. The Beltrami coefficient associated with the composition map can be expressed in terms of the Beltrami coefficients of the two quasiconformal maps.
% \begin{thm}\label{thm:composition}
% Let $f: \Omega \subset \mathbb{C} \to f(\Omega)$ and $g: f(\Omega) \to \mathbb{C}$ be two quasiconformal maps. The Beltrami coefficient of $g \circ f$ is explicitly given by
% \begin{equation}\label{eqt:composition_formula}
% \mu_{g \circ f} = \frac{\mu_f+(\overline{f_z}/f_z) (\mu_g \circ f)}{1+(\overline{f_z}/f_z)  \overline{\mu_f} (\mu_g \circ f)}.
% \end{equation}
% \end{thm}

Conversely, with a given complex function, a quasiconformal map can also be computed. More specifically, given a Beltrami coefficient $\mu:\mathbb{C}\to \mathbb{C}$ with $\|\mu\|_\infty < 1$, there exists a quasiconformal map satisfying the Beltrami equation (\ref{eqt:beltrami}) in the distribution sense \cite{Gardiner00}. 

To explicitly compute the quasiconformal map $f = u+iv$ with the given Beltrami coefficient $\mu = \rho + i \tau$, note that from the Beltrami Equation (\ref{eqt:beltrami}), each pair of the partial derivatives $v_x, v_y$ and $u_x, u_y$ can be expressed as linear combinations of the other \cite{Lui13}:
\begin{equation} \label{eqt:firstorder}
 \begin{split}
  v_y &= \alpha_1 u_x + \alpha_2 u_y; \\
  -v_x &= \alpha_2 u_x + \alpha_3 u_y,
 \end{split}
  \ \ \text{ and } \ \
 \begin{split}
  -u_y &= \alpha_1 v_x + \alpha_2 v_y; \\
  u_x &= \alpha_2 v_x + \alpha_3 v_y,
 \end{split}
\end{equation}
where $\alpha_1 = \frac{(\rho-1)^2+\tau^2}{1-\rho^2-\tau^2}; \alpha_2 = -\frac{2\tau}{1-\rho^2-\tau^2}; \alpha_3 = \frac{(1+\rho)^2 + \tau^2}{1-\rho^2-\tau^2}$. Since $\nabla \cdot \left( \begin{array}{c}
 -v_y \\ v_x
\end{array} \right) = 0$ and $\nabla \cdot \left( \begin{array}{c}
 -u_y \\ u_x
\end{array} \right) = 0$ , $f$ can be obtained by solving
\begin{equation}\label{eqt:secondorder}
\nabla \cdot \left(A \left(\begin{array}{c}
u_x\\
u_y \end{array}\right) \right) = 0\ \ \mathrm{and}\ \ \nabla \cdot \left(A \left(\begin{array}{c}
v_x\\
v_y \end{array}\right) \right) = 0
\end{equation}
where
$\displaystyle A = \left( \begin{array}{cc} \alpha_1 & \alpha_2\\
\alpha_2 & \alpha_3 \end{array}\right)$. Equation (\ref{eqt:secondorder}) is called the \emph{generalized Laplace equation}. %Consequently, to find an optimal quasiconformal map, it suffices to find an optimal complex-valued function.

% Quasiconformal maps can also be defined between two Riemann surfaces $\mathcal{M}$ and $\mathcal{N}$ by introducing the concept of Beltrami differentials. A \emph{Beltrami differential} $\mu(z)\frac{\overline{dz}}{dz}$ on $\mathcal{M}$ is an assignment to each chart $(U_{\alpha}, \phi_{\alpha})$ of an $L^{\infty}$ complex-valued function $\mu_{\alpha}$ defined on the local parameter $z_{\alpha}$ such that $\mu_{\alpha}(z_{\alpha}) \frac{ \overline{dz_\alpha}}{dz_\alpha} = \mu_{\beta}(z_{\beta}) \frac{ \overline{dz_{\beta}}}{dz_{\beta}}$ on the domain also covered by another chart $(U_\beta, \psi_\beta)$, where
% $\frac{d z_{\beta}}{d z_{\alpha}} = \frac{d}{d z_{\alpha}} \phi_{\alpha \beta}$ and $\phi_{\alpha \beta} = \phi_{\beta} \circ \phi_{\alpha}^{-1}$. An orientation preserving diffeomorphism $f: \mathcal{M} \to \mathcal{N}$ is said to be {\it quasiconformal} associated with $\mu(z)\frac{d\overline{z}}{dz}$ if for any chart $(U_{\alpha}, \phi_{\alpha})$ on $\mathcal{M}$ and any chart $(U_\beta, \psi_\beta)$ on $\mathcal{N}$, the composition map $f_{\alpha \beta} := \psi_{\beta} \circ f \circ f_{\alpha}^{-1}$ is quasiconformal associated with $\mu_{\alpha} \frac{ \overline{dz_{\alpha}}}{dz_{\alpha}}$.

In the discrete case, the Beltrami coefficients can be approximated on every triangular face. Let $f:K_1 \to K_2$ be a quasiconformal map between two triangulated meshes $K_1, K_2$, and let $T_1, T_2$ be two corresponding faces on $K_1, K_2$ respectively. Suppose $T_1 = [a_1+i\ b_1, a_2+ i\ b_2, a_3 + i\ b_3]$ and $T_2 = [w_1, w_2, w_3]$, where $a_i, b_i \in \mathbb{R}$ for all $i$. The Beltrami coefficient of $f$ is approximated on $T_1$ by
\begin{equation}\label{eqt:discreteBC}
\mu_f(T_1) = \frac{\frac{1}{2} \left(D_x + i\ D_y \right) \left( \begin{array}{c} w_1 \\ w_2 \\ w_3 \end{array} \right) }{\frac{1}{2} \left(D_x - i\ D_y \right) \left( \begin{array}{c} w_1 \\ w_2 \\ w_3 \end{array} \right) },
\end{equation}
where \begin{equation}
D_x = \frac{1}{2 Area(T_1)} \left( \begin{array}{c} b_3-b_2 \\ b_1-b_3 \\ b_2-b_1 \end{array} \right)^t \ \ \mathrm{and}\ \
D_y = -\frac{1}{2 Area(T_1)} \left( \begin{array}{c} a_3-a_2 \\ a_1-a_3 \\ a_2-a_1 \end{array} \right)^t.
\end{equation}

Similarly, $\alpha_1, \alpha_2, \alpha_3$ in Equation (\ref{eqt:firstorder}) can be discretized. Ultimately, the elliptic PDEs (\ref{eqt:secondorder}) can be discretized into sparse symmetric positive definite linear systems as described in the Linear Beltrami Solver (LBS) method \cite{Lui13}.

% Before ending this section, we introduce a special class of quasiconformal maps, called the Teichm\"uller maps.
% 
% \begin{defn}[Teichm\"uller map]
% Let $f:\mathcal{M} \to \mathcal{N}$ be a quasiconformal map between two Riemann surfaces $\mathcal{M}$ and $\mathcal{N}$. $f$ is said to be a \emph{Teichm\"uller map (T-map)} associated with the quadratic differential $q = \varphi dz^2$ where $\varphi:\mathcal{M} \to \mathbb{C}$ is a holomorphic function, if its associated Beltrami coefficient is of the form
% \begin{equation}\label{Teichmullermap}
% \mu(f) = k \frac{\overline{\varphi}}{|\varphi|}
% \end{equation}
% for some constant $k < 1$ and quadratic differential $q \neq 0 $ with $||q||_1 = \int_{S_1} |\varphi| <\infty$.
% \end{defn}
% 
% It follows that a Teichm\"uller map is a quasiconformal map whose Beltrami coefficient has a constant norm.

It is noteworthy that the focus in this section is only the quasiconformal maps on the complex plane. Nevertheless, since in this work we only consider genus-0 closed surfaces, which are conformally equivalent to $\mathbb{S}^2$ and hence the extended complex plane, the above concepts and discretizations can be naturally extended for our study.

\section{Our proposed method} \label{main}
% define qc distortion on surface
% K-qc
Note that quasiconformal maps are flexible and not unique in general. Therefore, it is desirable to have an algorithm for computing a spherical quasiconformal parameterization based on an user-defined quasiconformal distortion. The user-defined distortion can be freely set in order to fit into different applications. To achieve this goal, we first develop a measurement of qausiconformal distortion. Then, we propose a fast algorithm to compute a spherical quasiconformal parameterization with a given distortion.

\subsection{Quasiconformal dilation}
It is desirable to have a quantity that accurately represents the quasiconformality and is easy to compute. For spherical conformal maps, it is common to use the angle difference between the three angles of a triangular face on the input mesh and those of the face on the sphere as a measure of the conformality. Specifically, a map is with good conformality on the face if the three angle differences are all close to 0, or equivalently, if the mean of the absolute angle differences is close to 0. However, the measurement is not appropriate for the case of spherical quasiconformal maps. For instance, under the shear mapping $\begin{pmatrix}                                                                                                                                                                                                                                                                                                                                                                                                                                                                                                                                                                             x \\ y                                                                                                                                                                                                                                                                                                                                                                                                                                                                                                                                                  \end{pmatrix} \mapsto \begin{pmatrix}                                                                                                                                                                                                                                                                                                                                                                                                                                                                                                                                                                             x+ \lambda y \\ y                                                                                                                                                                                                                                                                                                                                                                                                                                                                                                                                                  \end{pmatrix} 
$, the three angle differences are highly different from each other and none of them can accurately represent the quasiconformality or the level of the distortion. Hence, instead of the angles, it is desirable to have the user-defined distortion defined on every triangular face of the input mesh.

In the following, we consider the \emph{dilation} on every triangular face as a measurement of quasiconformality. Mathematically, let $f:\mathbb{C} \to \mathbb{C}$ be a quasiconformal map. The \emph{dilation} of $f$ at a point $z$ is defined by
\begin{equation} 
K_f (z) = \frac{1+|\mu_f(z)|}{1-|\mu_f(z)|},
\end{equation}
where $\mu_f$ is the Beltrami coefficient of $f$. Geometrically, the dilation is the ratio of the length of the axes shown in Figure \ref{fig:qcmap} under the quasiconformal map $f$. 

The dilation of $f$ is related to the maximal quasiconformal dilation $K$ in Equation (\ref{eqt:maximalqcdilation}). More specifically, we have
\begin{equation}
K = \sup_z K_f (z).
\end{equation}

The map $f$ is said to be \emph{$p$-quasiconformal} if the maximal quasiconformal dilation is bounded above by $p$. In other words, every infinitesimal circle is mapped to an infinitesimal ellipse with eccentricity at most $p$. In particular, a conformal map is a $1$-quasiconformal map.

An important property about the maximal dilation of composition of quasiconformal mappings is as follows. 
\begin{prop} \label{prop:dilation}
If $f: \Omega_1 \to \Omega_2$ is a $K_1$-quasiconformal map and $g: \Omega_2 \to\Omega_3$ is a $K_2$-quasiconformal map, then $g \circ f$ is a $K_1K_2$-quasiconformal map. 
\end{prop}

In the discrete case, since the Beltrami coefficients are approximated on every triangular face as described in Equation (\ref{eqt:discreteBC}), it is natural to define the dilation on every face. We have the following discretization:

\begin{defn}[Discrete dilation]
Let $f:M_1 \to M_2$ be a quasiconformal map between two triangulated meshes $M_1, M_2$ on $\mathbb{C}$. For every triangular face $T$ of $M_1$, the discrete dilation of $f$ on $T$ is defined by
\begin{equation}\label{eqt:qcdilation}
 K_f (T) = \frac{1+|\mu_f (T)|}{1-|\mu_f(T)|},
\end{equation}
where $\mu_f(T)$ is the Beltrami coefficient of $f$ approximated on $T$.
\end{defn}

Moreover, the measurement of the dilation can be naturally extended to quasiconformal maps between meshes in $\mathbb{R}^3$.

\begin{defn}[Discrete dilation in $\mathbb{R}^3$]
Let $f:M_1 \to M_2$ be a quasiconformal map between two triangulated meshes $M_1, M_2$ in $\mathbb{R}^3$, and let $T_1, T_2$ be two corresponding triangular faces on $K_1, K_2$ respectively. Let $\phi_i: T_i \to \mathbb{C}$ be an isometric embedding of $T_i$ onto $\mathbb{C}$, where $i = 1,2$. The discrete dilation of $f$ on $T_1$ is defined by
\begin{equation}
K_{\tilde{f}} (\phi_1(T_1)),
\end{equation}
where $\tilde{f}: \phi_1(T_1) \to \phi_2(T_2)$ is a quasiconformal map on $\mathbb{C}$.
\end{defn}
Note that the above definition is well-defined because only the norm of the Beltrami coefficients is considered. With the above concepts, we are ready to introduce our proposed spherical quasiconformal parameterization algorithm for a genus-0 closed triangulated mesh $M$ and a user-defined quasiconformal dilation $K \geq 1$ defined on every face.

\subsection{Initial map}
We first compute a spherical conformal parameterization $f:M \to \mathbb{S}^2$ as an initialization. Among all existing algorithms for computing the spherical conformal parameterization, we choose the fast spherical conformal parameterization algorithm in \cite{Choi15a} for three reasons. Firstly, the algorithm only involves solving two sparse linear systems and hence the computation is highly efficient. Secondly, the algorithm in \cite{Choi15a} achieves the best conformality when compared with the existing approaches. The conformality of the initial spherical map is important in the subsequent steps. Thirdly, the algorithm in \cite{Choi15a} results in a bijective spherical parameterization. The bijectivity is also crucial for the computation in the remaining steps. 

\subsection{Optimally projecting the sphere onto the complex plane} \label{regular}
After obtaining the initial spherical parameterization, we choose a triangular face $T=[v_1,v_2,v_3]$ on $f(M)$ such that $T$ and its neighboring triangular faces are the most regular. Then, we apply a rotation $\psi$ on $f(M)$ such that the centroid of $T$ lies on the positive $z$-axis, followed by the stereographic projection $P_N$.

The regularity of $T$ and its neighboring faces is important because of the stereographic projection $P_N$. When applying the stereographic projection, the north pole $(0,0,1)$ is mapped to $\infty$ on the extended complex plane, and the northernmost region on $\mathbb{S}^2$ is mapped to the outermost region on the plane. In particular, $T$ is mapped to a big triangle on the plane. Now, denote the geodesic between $v_i$ and $v_j$ on $\mathbb{S}^2$ by $g_{v_i v_j}$. Note that $g_{v_i v_j}$ is a circular arc on $\mathbb{S}^2$, while the edge $e_{v_i v_j}$ connecting $v_i$ and $v_j$ on $M$ is an Euclidean straight line. On $\mathbb{S}^2$, this discrepancy between $g_{v_i v_j}$ and $e_{v_i v_j}$ may not be very large. However, under the stereographic projection, this discrepancy between $P_N(g_{v_i v_j})$ and the Euclidean straight line $e_{P_N(v_i) P_N(v_j)}$ becomes serious. 

In the continuous case, under the stereographic projection, all other vertices are mapped to the interior of the region enclosed by $g_{v_1 v_2}$, $g_{v_2 v_3}$ and $g_{v_3 v_1}$. However, in the discrete case, if $T$ and its neighboring faces are not regular enough, some vertices may be mapped outside the Euclidean triangle $[P_N(v_1), P_N(v_2), P_N(v_3)]$. The outlying vertices causes computational difficulty in the following step, in which only the three vertices $P_N(v_1), P_N(v_2), P_N(v_3)$ are involved in the boundary constraints. Hence, a suitable choice of $T$ is necessary.

\subsection{Achieving the desired quasiconformality}

By the stereographic projection, the chosen triangular face $T$ is mapped to a big triangle on $\mathbb{C}$. Next, we compose the map with a quasiconformal map $h$ that satisfies the prescribed dilation. 

To compute a quasiconformal map using LBS \cite{Lui13}, 3 point boundary constraints of the outermost triangular face $T$ are required. Moreover, the boundary constraints must be set optimally, otherwise the prescribed quasiconformality cannot be achieved. More specifically, the target location of the boundary points of $T$ should satisfy the prescribed quasiconformal dilation $K(T)$.

To explicitly compute the image of $T$ under the prescribed dilation $K(T)$, we denote $T = [x_1 + i y_1, x_2 + i y_2, x_3 + i y_3]$. By Equation (\ref{eqt:qcdilation}), we define the Beltrami coefficient $\mu(T)$ on the triangular face $T$ by
\begin{equation}
 \mu(T) = \frac{K(T)-1}{K(T)+1}.
\end{equation}
Note that the argument of $\mu(T)$ is set to be 0 without loss of generality. 

Since $h$ is piecewise linear, we have
\begin{equation} \label{eqt:piecewiselinear}
 h|_T \begin{pmatrix} x_i \\ y_i \end{pmatrix}  
 = \begin{pmatrix}
               a_T x_i + b_T y_i + r_T \\
               c_T x_i + d_T y_i + s_T
              \end{pmatrix}
\end{equation}
for $i=1,2,3$, where $a_T, b_T, c_T, d_T, r_T, s_T$ are to be determined.

Without loss of generality, we can assume that $h|_T \begin{pmatrix} x_1 \\ y_1 \end{pmatrix} = \begin{pmatrix} x_1 \\ y_1 \end{pmatrix}$ and $h|_T \begin{pmatrix} x_2 \\ y_2 \end{pmatrix} = \begin{pmatrix} x_2 \\ y_2 \end{pmatrix}$.

Also, by Equation (\ref{eqt:firstorder}), we have
\begin{equation}\label{eqt:firstorder_derive}
  \begin{split}
  d_T &= \alpha_1 a_T + \alpha_2 b_T; \\
  -c_T &= \alpha_2 a_T + \alpha_3 b_T,
 \end{split}
\end{equation}
where 
\begin{equation}
\alpha_1 = \frac{(\rho_T-1)^2+\tau_T^2}{1-\rho_T^2-\tau_T^2}; \ \ \ \alpha_2 = -\frac{2\tau_T}{1-\rho_T^2-\tau_T^2}; \ \ \ \alpha_3 = \frac{(1+\rho_T)^2 + \tau_T^2}{1-\rho_T^2-\tau_T^2}.
\end{equation}
Here, $\rho(T)$ and $\tau(T)$ are respectively the real part and the imaginary part of $\mu(T)$. By our construction of $\mu(T)$ introduced before, we have $\rho(T) = \frac{K(T)-1}{K(T)+1}$ and $\tau(T) = 0$. Hence, we have %A direct calculation yields
\begin{equation}
 \alpha_1 = \frac{\left(\frac{K(T)-1}{K(T)+1}-1\right)^2+0^2}{1-\left(\frac{K(T)-1}{K(T)+1}\right)^2-0^2} = \frac{1-\frac{K(T)-1}{K(T)+1}}{1+\frac{K(T)-1}{K(T)+1}} = -\frac{1}{K(T)}.
\end{equation}
Obviously, 
\begin{equation}
 \alpha_2 = 0.
\end{equation}
Lastly, we have
\begin{equation}
 \alpha_3 = \frac{\left(1+\frac{K(T)-1}{K(T)+1}\right)^2+0^2}{1-\left(\frac{K(T)-1}{K(T)+1}\right)^2-0^2} = \frac{1+\frac{K(T)-1}{K(T)+1}}{1-\frac{K(T)-1}{K(T)+1}} = K(T).
\end{equation}

% \begin{equation}
% \left\{ \begin{array}{l}
% \alpha_1 = -\frac{1}{K(T)} \\ \alpha_2 = 0 \\ \alpha_3 = K(T).
% \end{array}\right.
% \end{equation}

Altogether, $a_T, b_T, c_T, d_T, r_T, s_T$ can be explicitly solved by the following linear system:
\begin{equation} \label{eqt:bdy}
 \begin{pmatrix}
  x_1 & y_1 & 0 & 0 & 1 & 0 \\
  0 & 0 & x_1 & y_1 & 0 & 1 \\
  x_2 & y_2 & 0 & 0 & 1 & 0 \\
  0 & 0 & x_2 & y_2 & 0 & 1 \\
  \frac{1}{K(T)} & 0 & 0 & -1 & 0 & 0 \\
  0 & K(T) & 1 & 0 & 0 & 0 \\
 \end{pmatrix} 
 \begin{pmatrix}
  a_T \\ b_T \\ c_T \\ d_T \\ r_T \\ s_T \\
 \end{pmatrix} = 
 \begin{pmatrix}
  x_1 \\ y_1 \\ x_2 \\ y_2 \\ 0 \\ 0 \\
 \end{pmatrix}.
\end{equation}
Here, the first four equations come from Equation (\ref{eqt:piecewiselinear}), and the last two equations come from Equation (\ref{eqt:firstorder_derive}). The existence and uniqueness of $(a_T, b_T, c_T, d_T, r_T, s_T)$ is guaranteed by the following proposition.
\begin{prop}
The matrix in Equation (\ref{eqt:bdy}) is nonsingular.
\end{prop}
\begin{proof}
By a direct calculation, we have
\begin{align}
\det \begin{pmatrix}
  x_1 & y_1 & 0 & 0 & 1 & 0 \\
  0 & 0 & x_1 & y_1 & 0 & 1 \\
  x_2 & y_2 & 0 & 0 & 1 & 0 \\
  0 & 0 & x_2 & y_2 & 0 & 1 \\
  \frac{1}{K(T)} & 0 & 0 & -1 & 0 & 0 \\
  0 & K(T) & 1 & 0 & 0 & 0 \\
 \end{pmatrix}  
 &= \det \begin{pmatrix}
  x_1 & y_1 & 0 & 0 & 1 & 0 \\
  0 & 0 & x_1 & y_1 & 0 & 1 \\
  x_2-x_1 & y_2-y_1 & 0 & 0 & 0 & 0 \\
  0 & 0 & x_2-x_1 & y_2-y_1 & 0 & 0 \\
  \frac{1}{K(T)} & 0 & 0 & -1 & 0 & 0 \\
  0 & K(T) & 1 & 0 & 0 & 0 \\
 \end{pmatrix} \\
 & = \det \begin{pmatrix}
  x_2-x_1 & y_2-y_1 & 0 & 0 \\
  0 & 0 & x_2-x_1 & y_2-y_1 \\
  \frac{1}{K(T)} & 0 & 0 & -1 \\
  0 & K(T) & 1 & 0 \\
 \end{pmatrix} \\
 &= - K(T) (x_2 - x_1)^2 - \frac{1}{K(T)} (y_2-y_1)^2.
\end{align}

Since $T$ is non-degenerate, we have $(x_1,y_1) \neq (x_2, y_2)$. Also, note that $K \geq 1$. It follows that 
\begin{equation}
\det \begin{pmatrix}
  x_1 & y_1 & 0 & 0 & 1 & 0 \\
  0 & 0 & x_1 & y_1 & 0 & 1 \\
  x_2 & y_2 & 0 & 0 & 1 & 0 \\
  0 & 0 & x_2 & y_2 & 0 & 1 \\
  \frac{1}{K(T)} & 0 & 0 & -1 & 0 & 0 \\
  0 & K(T) & 1 & 0 & 0 & 0 \\
 \end{pmatrix}   \neq 0. 
 \end{equation}
\end{proof}

After obtaining $a_T, b_T, c_T, d_T, r_T, s_T$, we can explicitly compute $h|_T \begin{pmatrix} x_3 \\ y_3 \end{pmatrix}$ using Equation (\ref{eqt:piecewiselinear}). The above computations give us the desired boundary condition for $h(x_1 + i y_2), h(x_1 + i y_2)$ and $h(x_3 + i y_3)$ of the triangular face $T$. 

With the above boundary conditions, we apply the Linear Beltrami Solver (LBS) \cite{Lui13} for computing a quasiconformal map $h$ that satisfies the prescribed quasiconformal distortion. More specifically, by Equation (\ref{eqt:qcdilation}), we have 
\begin{equation}
 |\mu(F)| = \frac{K(F)-1}{K(F)+1}
\end{equation}
for all triangular faces $F$. We apply LBS with $\mu$ and the boundary constraints on $T$, obtaining the quasiconformal map $h$. It is noteworthy that since $\| \mu \|_{\infty} <1$, Theorem \ref{thm:bijectivity} guarantees the bijectivity of the map $h$.

Since $T$ may be severely distorted by the prescribed distortion, the origin may no longer be located inside $T$ under the quasiconformal map $h$. In this case, the resulting parameterization obtained by the inverse stereographic projection $P_S^{-1}$ may not be a sphere but only a portion of it. To overcome this problem, we perform a translation on $\mathbb{C}$ so that the centroid of the whole domain is at the origin. This ensures that $T$ will be the northernmost triangular face under $P_S^{-1}$.

Now, the desired quasiconformal distortion is achieved. However, as we have fixed two vertices of $T$ in computing the boundary constraints, the size of the whole triangular domain may not be optimal. More specifically, if the size of $T$ is too large, most vertices will be mapped to the northern hemisphere by $P_S^{-1}$. On the other hand, if the size of $T$ is too small, most vertices will be mapped to the southern hemisphere by $P_S^{-1}$. To achieve an optimal distribution on the spherical parameterization, we apply the balancing scheme in the fast spherical conformal parameterization algorithm \cite{Choi15a}. Based on Invariance Theorem in \cite{Choi15a}, the balancing scheme ensures that $T$ and the innermost triangle $t$ on $\mathbb{C}$ will be mapped to two triangles with similar size on the unit sphere under $P_S^{-1}$. This completes our task of computing a spherical quasiconformal parameterization with prescribed quasiconformal distortion.

It is noteworthy that our proposed algorithm only involves solving a few sparse linear systems. Hence, our algorithm is highly efficient in practice. Also, the desired quasiconformality of the spherical parameterization is guaranteed by Theorem \ref{prop:dilation}. Since the initial spherical map, the rotation and the stereographic projections are all conformal maps (i.e. $1$-quasiconformal maps) and $h$ is $K$-quasiconformal, the composition of the maps is also $K$-quasiconformal. Assembling all of the above steps, our proposed fast spherical quasiconformal (FSQC) parameterization algorithm is summarized in Algorithm \ref{algorithm}.

\begin{algorithm}[!h]
\KwIn{A genus-0 closed triangular mesh $M$, a user-defined quasiconformal dilation $K \geq 1$ defined on every face.}
\KwOut{A bijective spherical quasiconformal parameterization $\varphi:M \to \mathbb{S}^2$.}
\BlankLine
Compute a spherical conformal parameterization $f:M \to \mathbb{S}^2$ using the fast algorithm in \cite{Choi15a}\;
Choose a triangular face $T$ on $f(M)$ as described in Section \ref{regular}\;
Apply a rotation $\psi$ on $f(M)$ such that the centroid of $T$ lies on the positive $z$-axis\;
Apply the stereographic projection $P_N$ on $\psi(f(M)$\;
% Compute a quasiconformal map $g: P_N(\psi(f(M))) \to \mathbb{C}$ using LBS \cite{Lui13} with $\mu_g = \mu_{(P_N \circ \psi \circ f)^{-1}}$\; % to cancel the distortion due to projection\;
Compute a quasiconformal map $h: P_N(\psi(f(M))) \to \mathbb{C}$ with the prescribed distortion, and an appropriate boundary condition of the big triangle $T$\;
Perform a translation so that the centroid of the whole domain is at the origin\;
Apply the balancing scheme in \cite{Choi15a}\;
Apply the inverse stereographic projection $P_N^{-1}$ and denote the overall result by $\varphi$\;
\caption{Fast spherical quasiconformal (FSQC) parameterization}
\label{algorithm}
\end{algorithm}

\subsection{Remeshing via FSQC}

\begin{figure}[!t]
 \centering
 \includegraphics[width=0.9\textwidth]{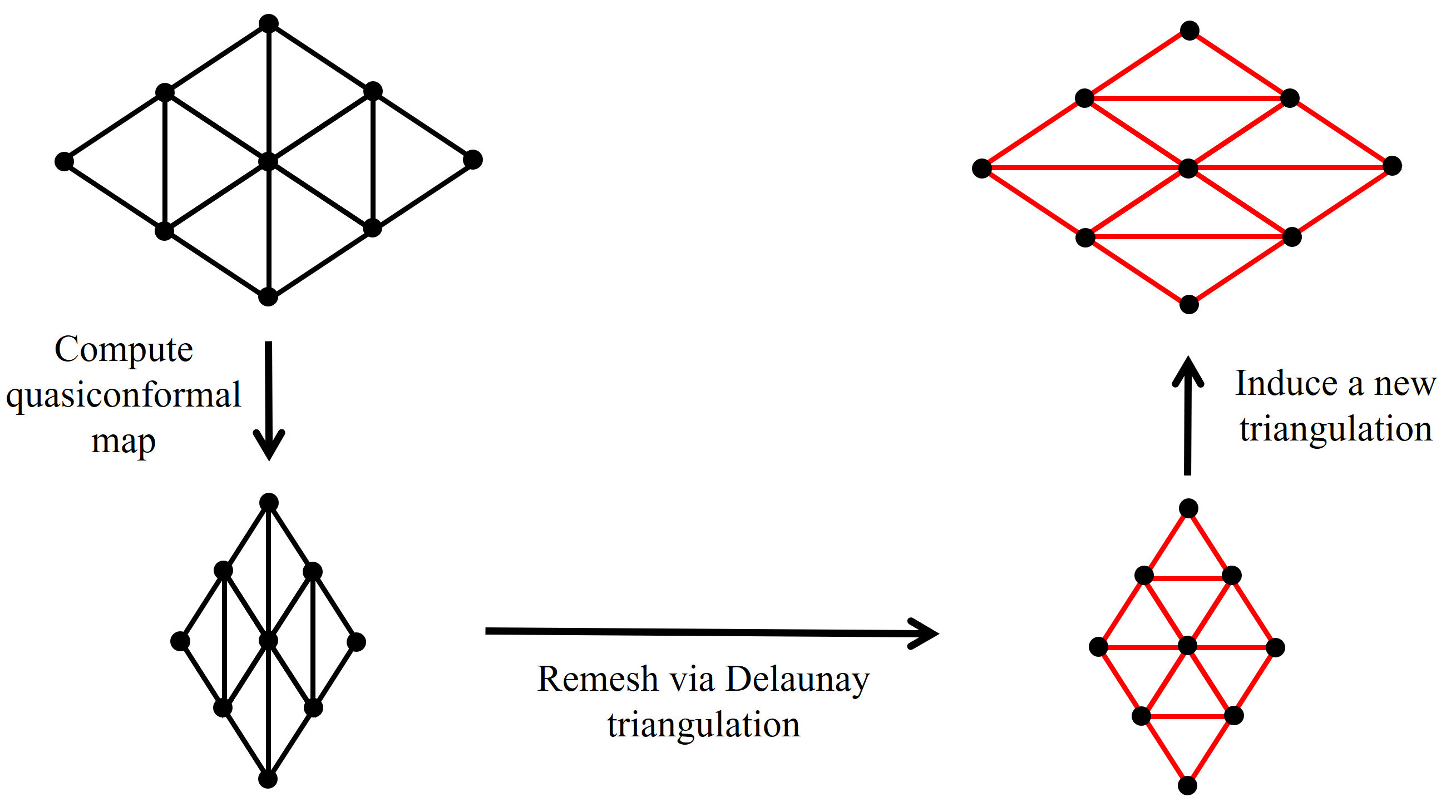}
 \caption{A simplified illustration of our proposed remeshing framework. To adaptive remesh a portion of a genus-0 closed surface, we can set a special quasiconformal dilation at that region and compute a spherical quasiconformal parameterization using FSQC. Then, we can apply the spherical Delaunay triangulation algorithm to remesh the corresponding region on the spherical parameterization. Finally, the resulting triangulation on the sphere induces a new triangulation of the original surface.}
 \label{fig:remeshing_procedure}
\end{figure}

The spherical qausi-conformal parameterization obtained by our FSQC algorithm can be used for remeshing an input genus-0 closed mesh $M = (V,F)$, where $V$ is the set of vertices and $F$ is the set of triangular faces of $M$. This brief idea of our framework is that we can apply existing triangulation algorithms, such as the spherical Delaunay triangulation algorithm, for creating a triangulation on the spherical parameterization of $M$. Then, the spherical triangulation induces a triangulation $F'$ on $M$ and this completes the task of remeshing $M$. 

A simplified illustration of our proposed remeshing framework is given in Figure \ref{fig:remeshing_procedure}. The brief idea of our proposed remeshing framework is as follows. Suppose we have a set of points at a sharp part (for instance, a tail) of a genus-0 closed surface. Note that a regular triangulation of the set of points may not lead to a good visualization of the sharpness of the part. Instead, it is desirable to form sharp triangles on the set of points. To achieve this, we compute a sphere quasiconformal parameterization that squeezes the region. Then, we apply existing triangulation algorithms on the simple spherical domain to construct a regular triangulation. The regular triangulation built on the spherical domain induces a triangulation on the original surface. Because of the quasiconformal distortion, the induced triangulation is with sharp triangles at the mentioned region. Therefore, the new triangulation enhances the visual quality of the sharp part of the surface. In the following, we explain our proposed remeshing framework in details. 

We start by rigorously introducing the Delaunay triangulation. Mathematically, the definition of the Delaunay triangulation is as follows.
\begin{defn}[Delaunay triangulations]
A triangulation of a set of points $\mathcal{P}$ is said to be \emph{Delaunay} if for any triangle $T$ in the triangulation, no point in $\mathcal{P}$ lies inside the circumcircle of $T$. 
\end{defn}

Consequently, Delaunay triangulations avoid sharp triangles and produce as many regular triangles as possible. As described in \cite{Choi15d}, for the case of conformal parameterization, a regular triangulation on the spherical parameterization induces a regular triangulation on the original surface. On the contrary, for the case of quasiconformal parameterization, the induced triangulation may not be regular due to the quasiconformal distortion. Nevertheless, it is the discrepancy caused by quasiconformal distortion that enables us to adaptively remesh a surface. 

To apply our proposed FSQC algorithm for adaptive surface remeshing, we now describe a strategy in setting the user-defined quasiconformal dilations. 

Let $R$ be a simply-connected set of triangular faces on $M=(V,F)$ that we want to irregularize. We set the user-defined quasiconformal dilation to be  
\begin{equation}
 K(T) \left\{ \begin{array}{ll}
                  = 1 & \text{ if } T \in F \setminus R,\\
                  \gg 1 & \text{ if } T \in R.
                \end{array}\right.
\end{equation}
Then, we select two vertices $p_1, p_2$ that represent the principal direction of the region. This can be done manually or by existing methods such as the Principal Component Analysis (PCA).

As the quasiconformal dilation $K$ only encodes the magnitude but not the direction of the desired distortion, we need to insert one extra step in running our proposed FSQC algorithm. More specifically, after Step 4 in Algorithm \ref{algorithm}, we have obtained $(P_N \circ \psi \circ f)(M) \subset \overline{\mathbb{C}}$. Before proceeding to Step 5 in Algorithm \ref{algorithm}, we rotate the entire planar domain by a map
\begin{equation}
 z \mapsto z e^{i \theta}
\end{equation}
where
\begin{equation}
 \theta = \text{Arg}((P_N \circ \psi \circ f)(p_2) - (P_N \circ \psi \circ f)(p_1)).
\end{equation}

This step ensures that the highlighted region $R$ will be squeezed in a direction perpendicular to the the line joining $(P_N \circ \psi \circ f)(p_1)$ and $(P_N \circ \psi \circ f)(p_2)$. Then, we continue the FSQC algorithm and obtain the final spherical parameterization $\varphi:M \to \mathbb{S}^2$. With the spherical parameterization $\varphi(M)$, we can apply the spherical Delaunay triangulation algorithm on the vertices of $\varphi(M)$. The Delaunay triangulation obtained on $\varphi(M)$ induces a triangulation $F'$ on the original surface $M$. It is noteworthy that because of the user-defined quasiconformal dilation, the artificially expanded region on $\varphi(M)$ leads to the formulation of squeezed triangles in the remeshing result $F'$.

Our proposed remeshing framework is summarized in Algorithm \ref{algorithm_remeshing}.
% Besides, this approach can be extended for building a quad mesh. 

\begin{algorithm}[!h]
\KwIn{A genus-0 closed triangular mesh $M = (V,F)$.}
\KwOut{The remeshed surface $M' = (V,F')$.}
\BlankLine
Set the quasiconformal dilation $K=1$ for all triangular faces\;
Highlight a region to be adaptively remeshed and set $K \gg 1$ for the region, with $K = 1$ elsewhere\;
Select two vertices $p_1, p_2$ that represent the principal direction of the region\;
Apply our proposed FSQC with the quasiconformal dilation $K$, with an extra rotation of angle $\theta$ of all points on $\mathbb{C}$ right before Step 5 of Algorithm \ref{algorithm}. Here, $\theta = \text{Arg}((P_N \circ \psi \circ f)(p_2) - (P_N \circ \psi \circ f)(p_1))$\;
Apply the spherical Delaunay triangulation algorithm on the spherical parameterization\;
Obtain the induced triangulation $F$ from the spherical triangulation\;
\caption{Remeshing via FSQC}
\label{algorithm_remeshing}
\end{algorithm}

\section{Experimental Results} \label{experiment}

In this section, we demonstrate the effectiveness of our proposed fast spherical quasiconformal parameterization algorithm with application to adaptive surface remeshing. Various genus-0 closed triangulated meshes are adopted from the AIM@SHAPE Shape Repository \cite{aim@shape} and the Benchmark for 3D Mesh Segmentation \cite{Chen09} for testing our algorithm. Our algorithms are implemented in MATLAB. The spherical Delaunay triangulation algorithm in \cite{delaunay} is adopted for remeshing the spherical parameterizations. All experiments are performed on a PC with an Intel(R) Core(TM) i5-3470 CPU @3.20 GHz processor and 8.00 GB RAM.

\begin{figure}[!t]
 \centering
 \includegraphics[width=0.525\textwidth]{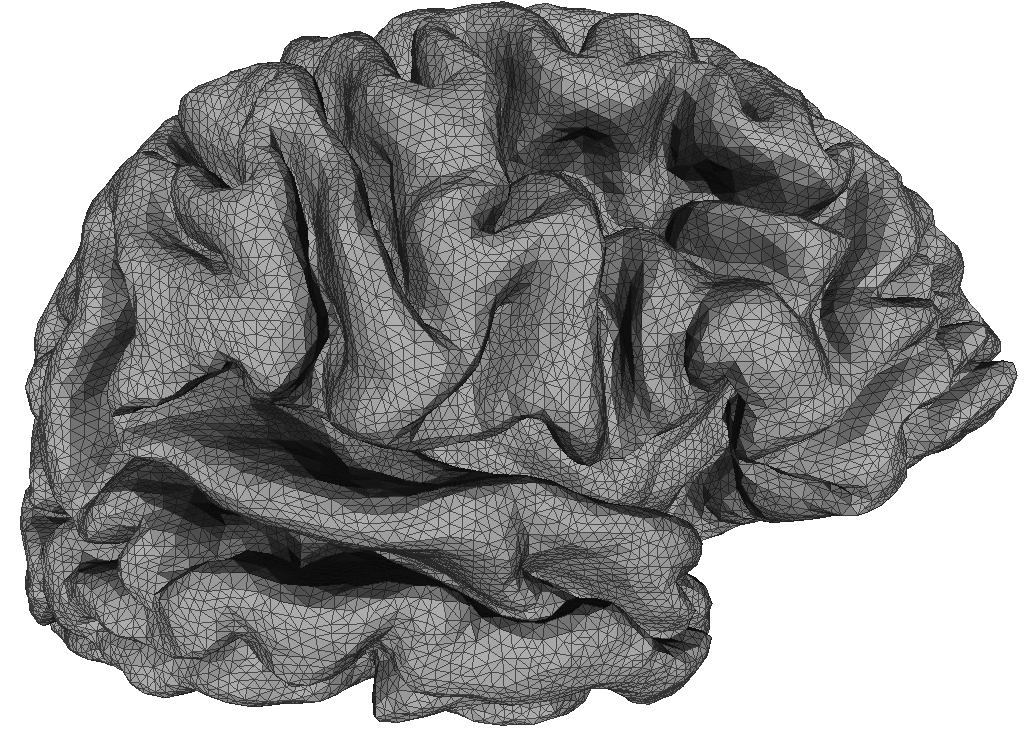}
 \includegraphics[width=0.375\textwidth]{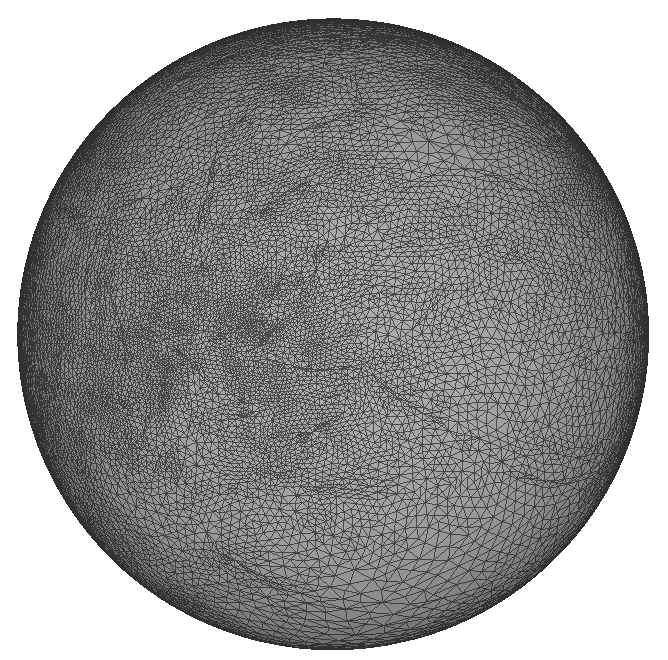} \\ \vspace{8mm}
 \includegraphics[width=0.45\textwidth]{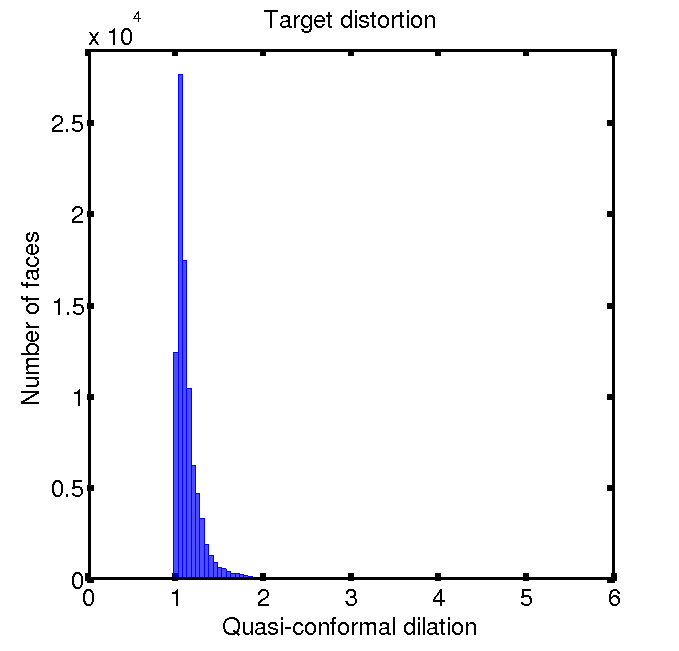}
 \includegraphics[width=0.45\textwidth]{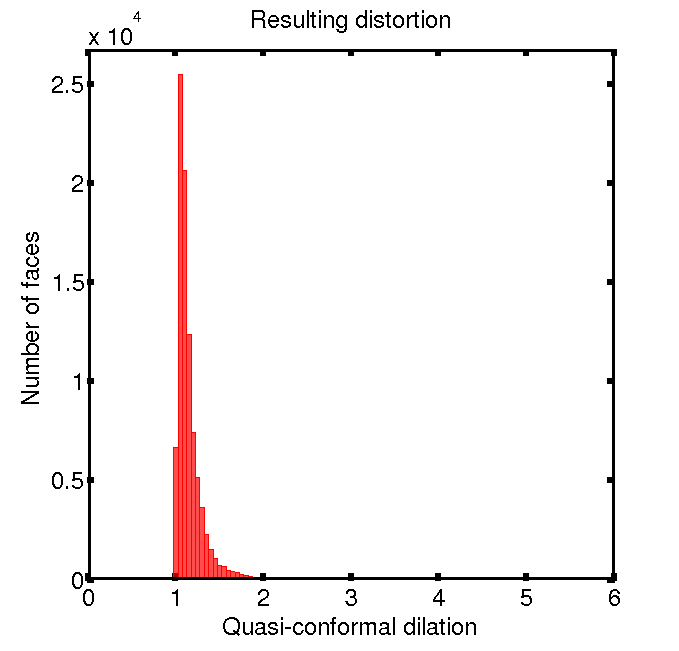}
 \caption{A brain and the spherical quasiconformal parameterization obtained by our algorithm. Top left: the input surface. Top right: the spherical parameterization. Bottom left: The target quasiconformal distortion. Bottom right: The resulting quasiconformal distortion of the parameterization.}
 \label{fig:brain}
\end{figure}

\begin{figure}[!t]
 \centering
 \includegraphics[width=0.35\textwidth]{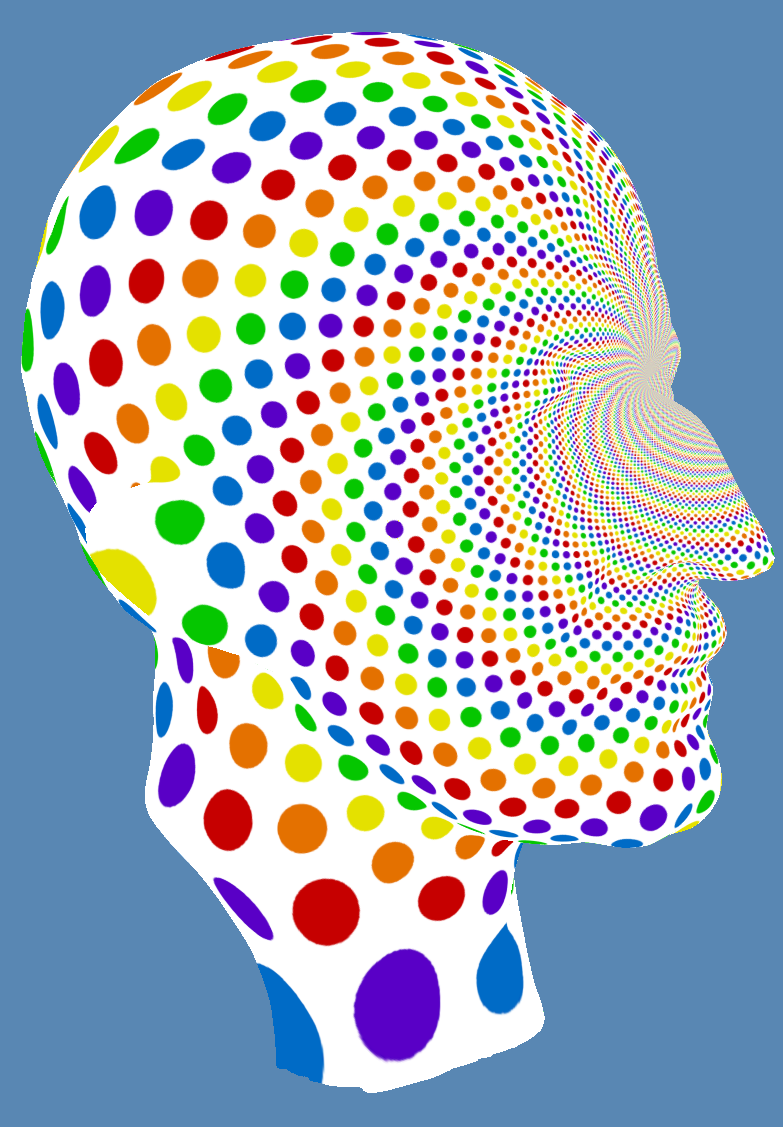}
 \includegraphics[width=0.51\textwidth]{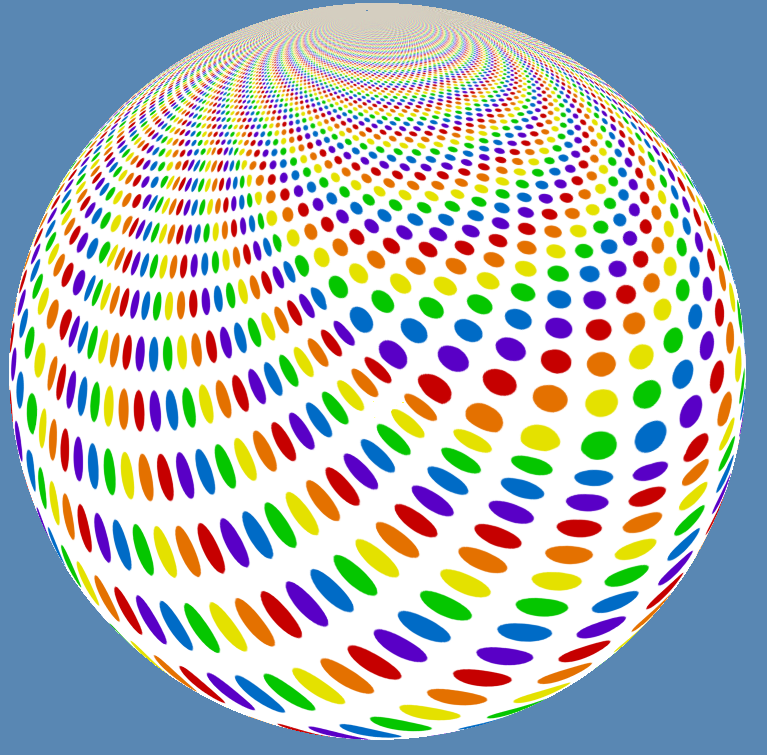}  \\ \vspace{8mm}
 \includegraphics[width=0.45\textwidth]{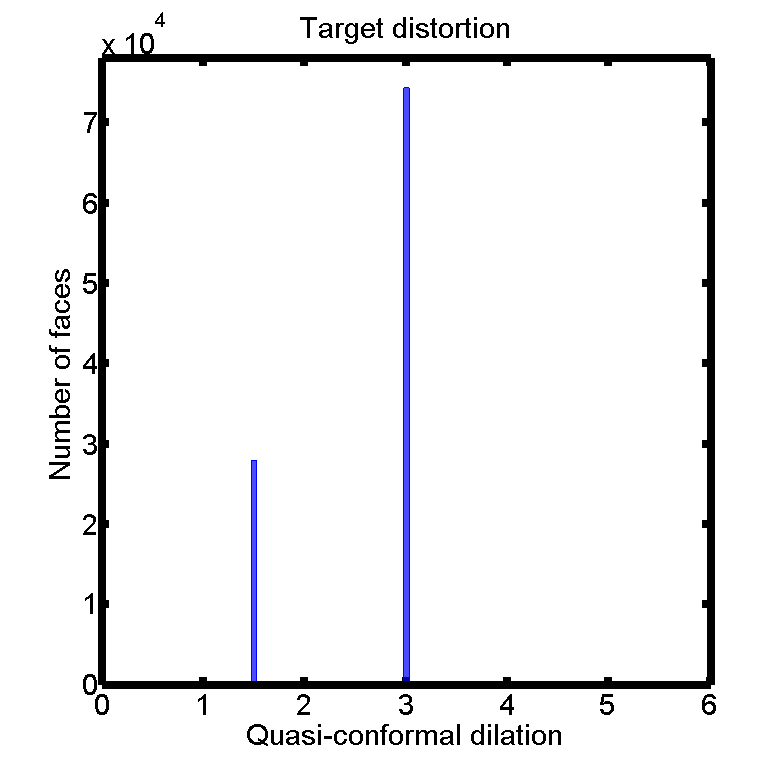}
 \includegraphics[width=0.45\textwidth]{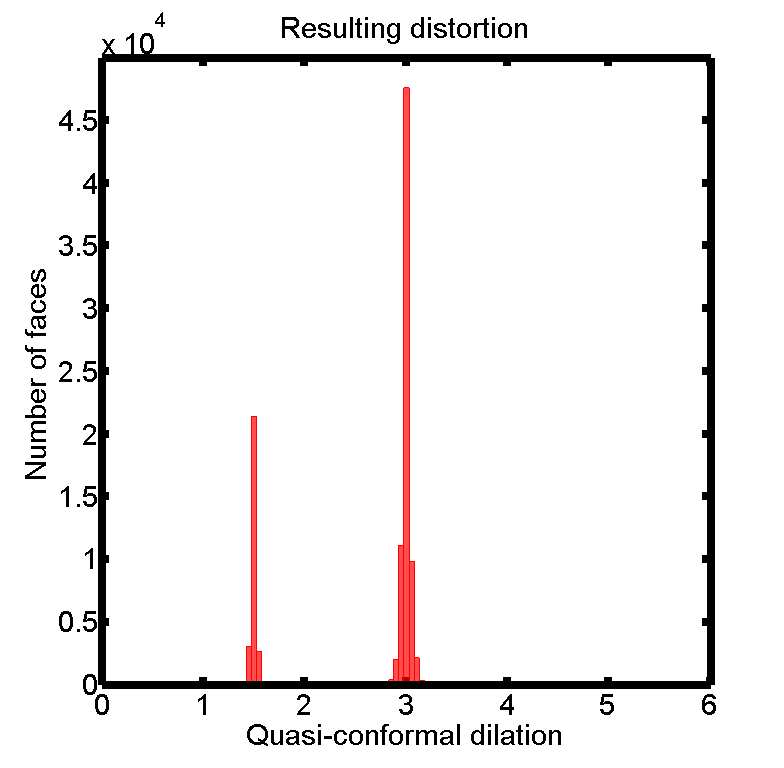}
 \caption{A Max Planck model with a circle pattern and the spherical quasiconformal parameterization obtained by our algorithm. Top left: the input surface. Top right: the spherical parameterization. Bottom left: The target quasiconformal distortion. Bottom right: The resulting quasiconformal distortion of the parameterization. The user-defined synthetic distortion is achieved on the spherical parameterization.}
 \label{fig:maxplanck}
\end{figure}

\begin{figure}[!t]
 \centering
 \includegraphics[width=0.45\textwidth]{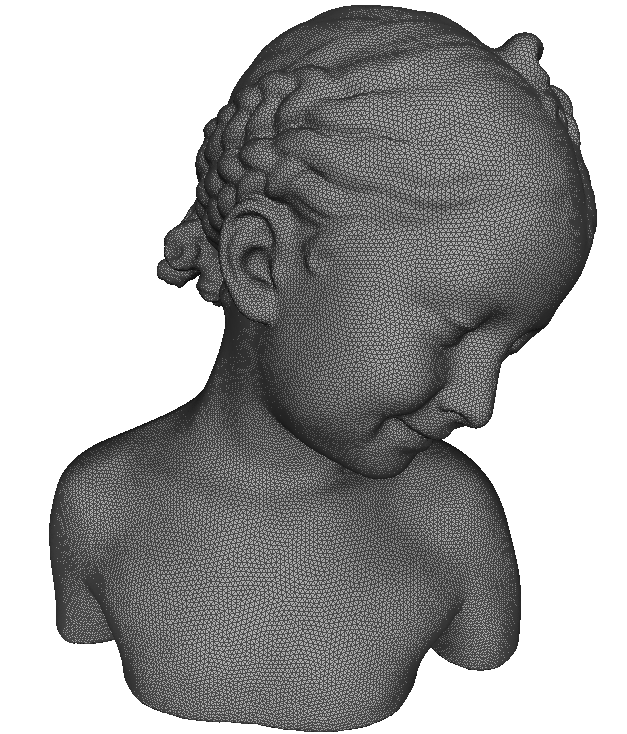}
 \includegraphics[width=0.45\textwidth]{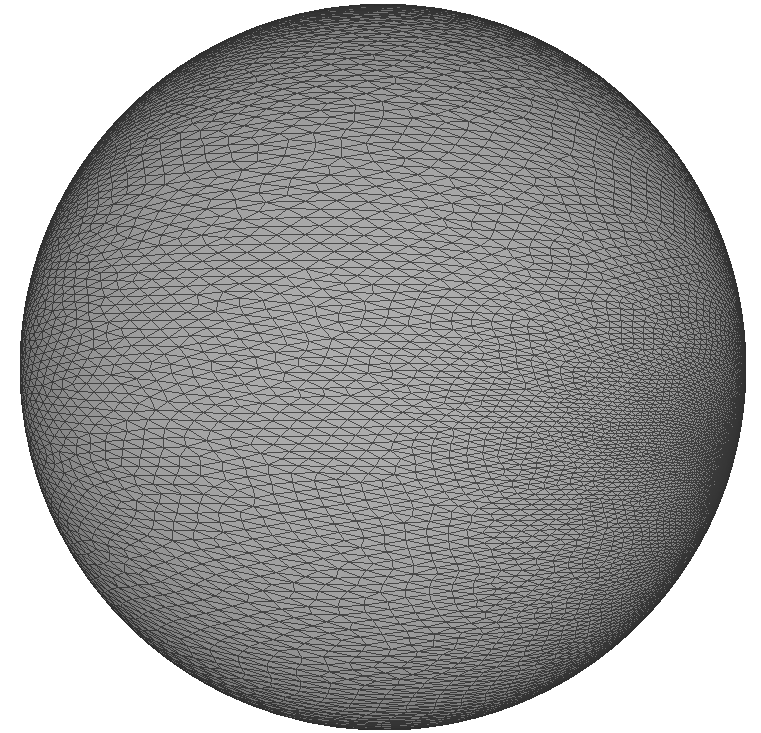}  \\ \vspace{8mm}
 \includegraphics[width=0.45\textwidth]{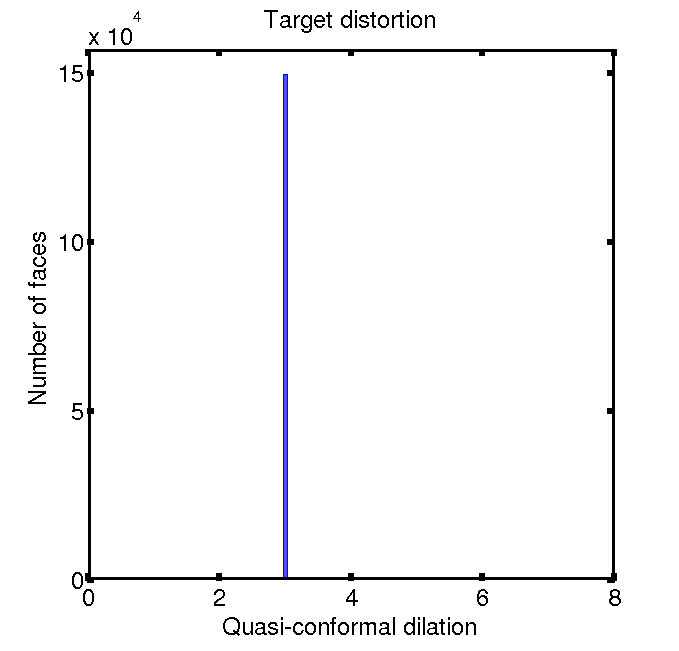}
 \includegraphics[width=0.45\textwidth]{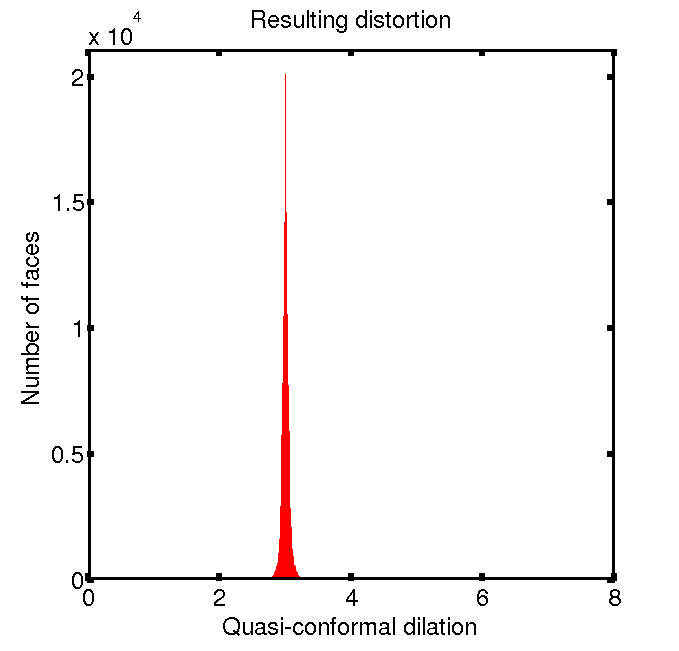}
 \caption{A bimba model and the spherical quasiconformal parameterization with uniform conformality distortion obtained by our proposed algorithm. Top left: the input surface. Top right: the spherical parameterization. Bottom left: The target quasiconformal distortion. Bottom right: The resulting quasiconformal distortion of the parameterization.}
 \label{fig:bimba}
\end{figure}

\begin{figure}[!t]
 \centering
 \includegraphics[width=0.45\textwidth]{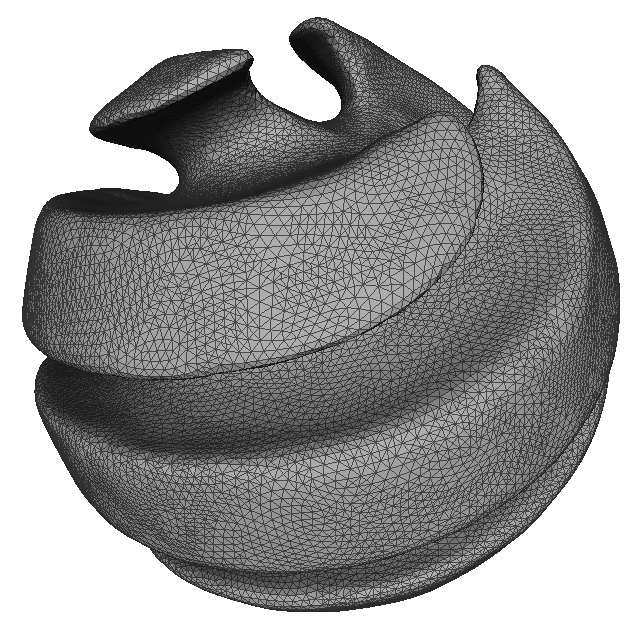}
 \includegraphics[width=0.45\textwidth]{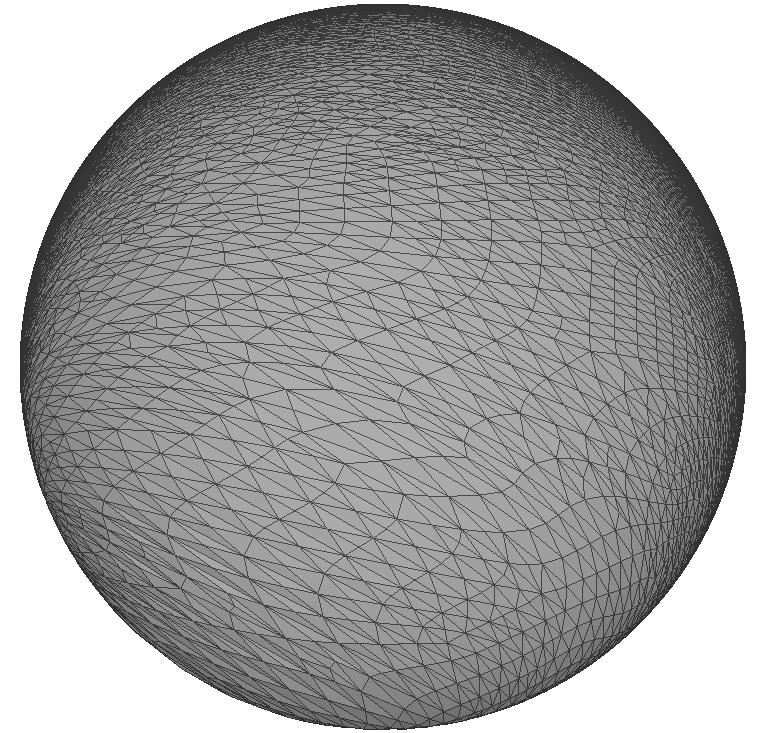} \\ \vspace{8mm}
 \includegraphics[width=0.45\textwidth]{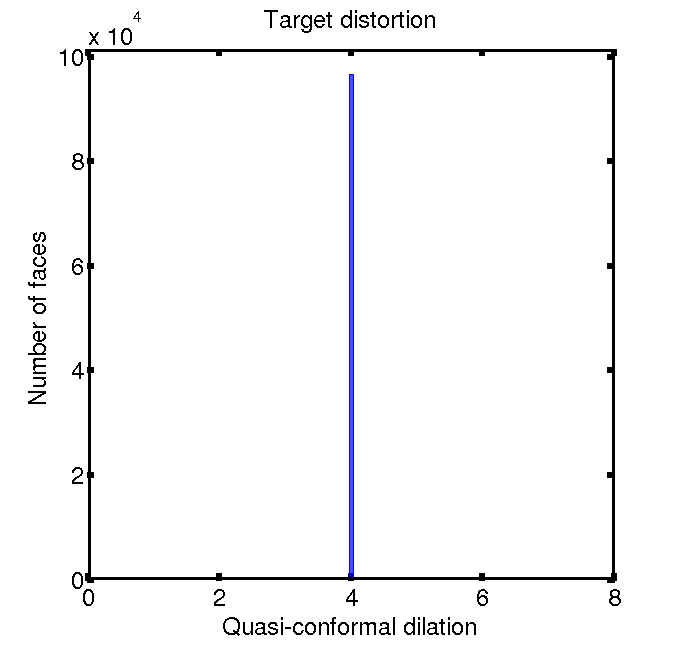}
 \includegraphics[width=0.45\textwidth]{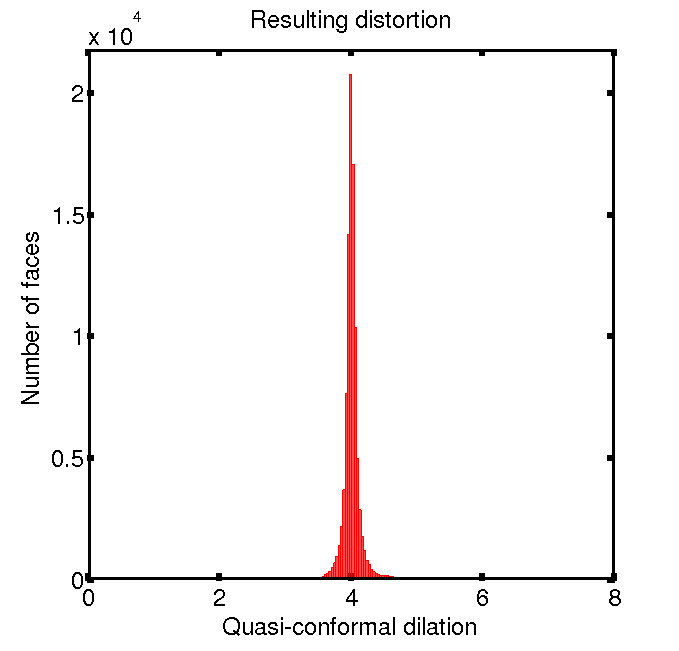}
 \caption{A spiral and the spherical quasiconformal parameterization with uniform conformality distortion obtained by our algorithm. Top left: the input surface. Top right: the spherical parameterization. Bottom left: The target quasiconformal distortion. Bottom right: The resulting quasiconformal distortion of the parameterization.}
 \label{fig:spiral}
\end{figure}

Figure \ref{fig:brain} shows a genus-0 closed brain mesh and the spherical quasiconformal parameterization obtained by our fast algorithm. It can be observed that the resulting quasiconformal distortion closely resembles the desired quasiconformal distortion. Another example is shown in Figure \ref{fig:maxplanck}. In this example, we consider a discontinuous dilation as the target quasiconformal distortion. Even with the discontinuity, the spherical quasiconformal parameterization obtained can satisfy the desired distortion. It can be observed that the circles on the input mesh are transformed to two types of ellipses on the spherical quasiconformal parameterization. Also, two sharp peaks can be observed in the histogram of the resulting dilation plot. 

Then, we apply our algorithm for computing spherical uniform conformality distortion parameterization of genus-0 closed meshes by setting the target dilation as a constant. Figure \ref{fig:bimba} and Figure \ref{fig:spiral} show two examples of the spherical uniform conformality distortion parameterizations obtained by our algorithm. It is noteworthy that even for the highly convoluted spiral model, the resulting dilations significantly concentrate at the desired constant. The uniform dilation can also be observed from the triangular faces on the spherical parameterizations. This implies that our algorithm can effectively produce the spherical parameterizations with uniform conformality distortion. 

Table \ref{table:sphericalqc} records the performance of our proposed fast spherical quasiconformal parameterization algorithm. Because of the sparse linear systems in our algorithm, the computations finish within a few seconds even for very dense meshes. Also, in all examples, the resulting quasiconformal distortion is highly close to the target distortion. This reflects the accuracy of our proposed algorithm. Besides, the absence of extreme values in the resulting dilation distribution implies that the Beltrami coefficient is with sup norm $\ll 1$. Hence, the resulting parameterizations are bijective.

\begin{table}[!t]
\footnotesize
    \begin{center}
    \begin{tabular}{ |l|c|c|C{15mm}|C{15mm}|C{15mm}|C{15mm}| }
    \hline
     Surfaces & \# of faces & Time (s) & \multicolumn{2}{ c| }{Target dilation} & \multicolumn{2}{ c| }{Resulting dilation} \\ \cline{4-7}
      &  &  & Mean & SD & Mean & SD \\ \hline
      Max Planck & 102212 & 1.8867 & 2.5887 & 0.6692 & 2.5896 & 0.6687 \\
      Brain 1 & 91124 & 1.9399 & 1.1496 & 0.2486 & 1.1643 & 0.2319 \\
      Brain 2 & 92210 & 2.0185 & 1.2149 & 0.3021 & 1.2246 & 0.3030 \\
      Lion & 100000 & 2.0651 & 1.2174 & 0.2180 & 1.2246 & 0.2228 \\
      Spiral & 96538 & 1.7577 & 4.0000 & 0.0000 & 4.0079 & 0.2552 \\
      Bimba & 149524 & 3.8332 & 3.0000 & 0.0000 & 3.0007 & 0.0704 \\ 
      Dolphin & 3784 & 0.0756 & 1.0876 & 0.3518 & 1.2766 & 0.3628 \\
      Human face & 43056 & 0.8036 & 1.0109 & 0.1202 & 1.0314 & 0.1190 \\
      \hline
    \end{tabular}
    \end{center}
    \bigbreak
    \caption{The performance of our fast spherical quasiconformal parameterization algorithm.}
    \label{table:sphericalqc}
\end{table}

\begin{figure}[!t]
 \centering
 \includegraphics[width=0.33\textwidth]{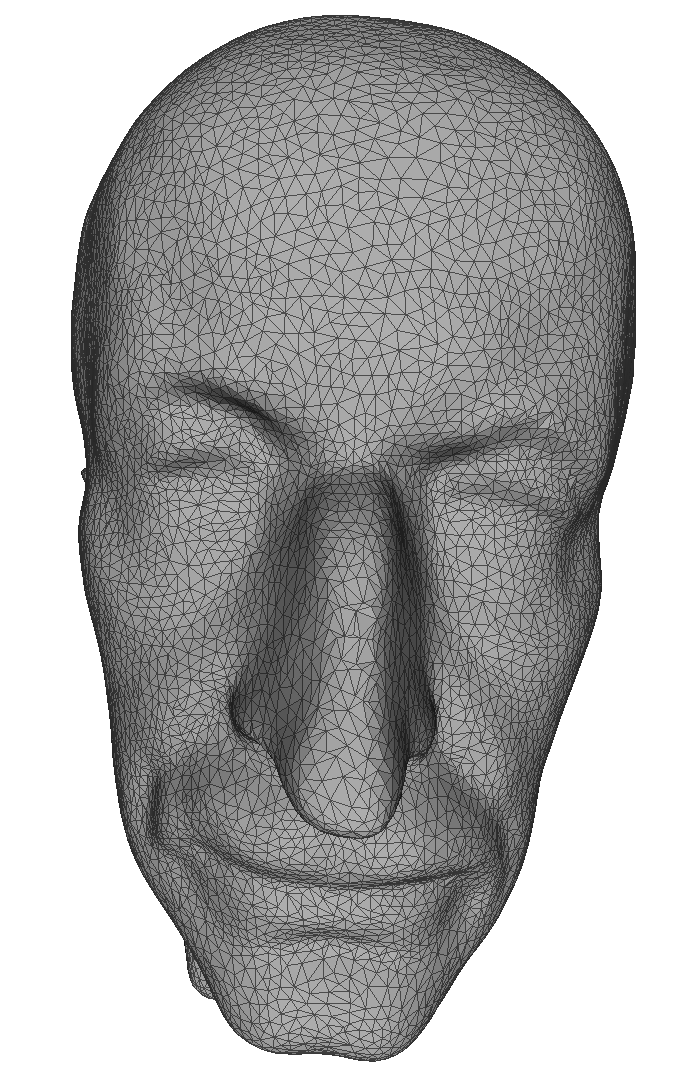} \ \ 
 \includegraphics[width=0.55\textwidth]{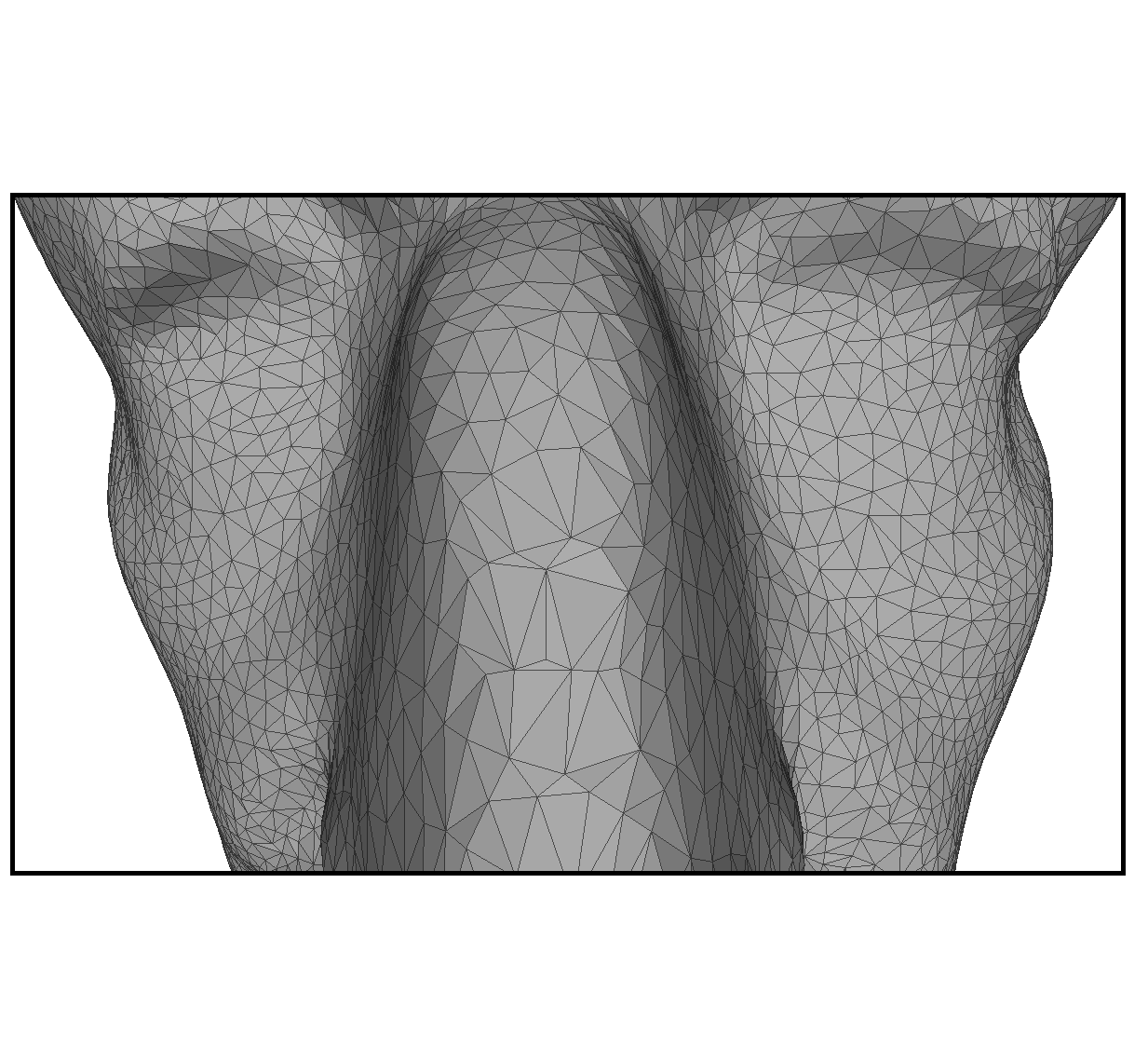}
 \caption{A human surface with a Delaunay triangulation. Note that the triangles at the nose bridge do not preserve the geometry well. Left: The whole surface. Right: A zoom-in of the nose bridge.}
 \label{fig:maxplanck_unremeshed}
\end{figure}

\begin{figure}[!t]
 \centering
 \includegraphics[width=0.33\textwidth]{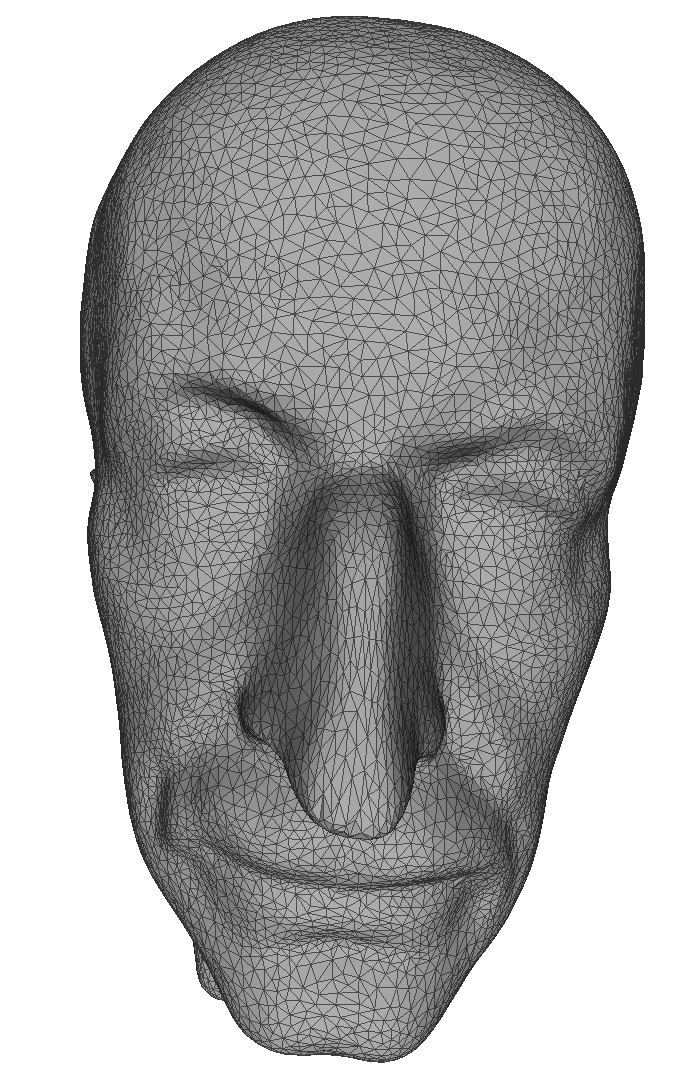} \ \ 
 \includegraphics[width=0.55\textwidth]{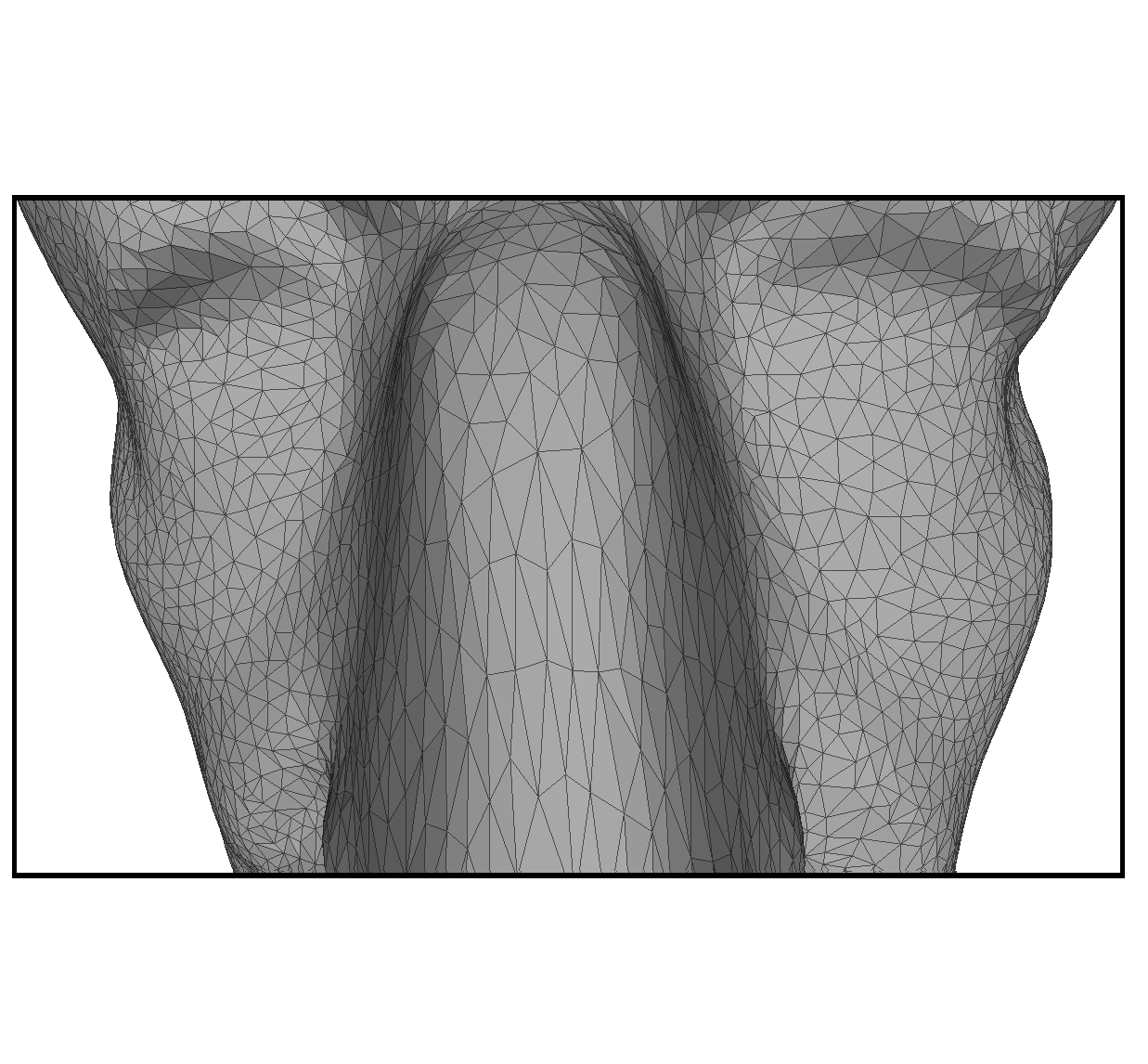}
 \caption{The remeshed human surface obtained by our proposed method. The sharp triangulations make nose bridge more prominent. Left: The whole surface. Right: A zoom-in of the nose bridge.}
 \label{fig:maxplanck_remeshed}
\end{figure}

After demonstrating the efficiency and accuracy of our proposed fast spherical quasiconformal parameterization algorithm, we apply the algorithm for adaptive remeshing. Figure \ref{fig:maxplanck_unremeshed} shows a human face represented by a Delaunay triangulation. Note that the triangular faces at the nose bridge do not follow the shape of the nose bridge and hence the nose bridge does not look prominent. We aim to remesh this particular part of the surface in order to enhance the visual quality. To achieve this, we set the quasiconformal dilation $K$ as
\begin{equation}
 K(T) = \left\{ \begin{array}{ll}
                  2.5 & \text{ if } T \text{ is at the nose bridge,}\\
                  1 & \text{ otherwise,}
                \end{array}\right.
\end{equation}
for all triangle elements $T$. Two points $p_1$, $p_2$ are manually selected at the top and the tip of the nose to control the direction of the distortion. Then, we apply our proposed FSQC algorithm with the quasiconformal dilation $K$ and obtain a spherical quasiconformal parameterization of the human surface. After that, we apply the spherical Delaunay triangulation algorithm to remesh the spherical parameterization. The final induced triangulation on the original surface is shown in Figure \ref{fig:maxplanck_remeshed}. Note that the triangulations at the nose bridge become sharp and naturally follow the geometry of the nose bridge. This improves the visualization of the human face.

\begin{figure}[!t]
 \centering 
 \includegraphics[width=0.42\textwidth]{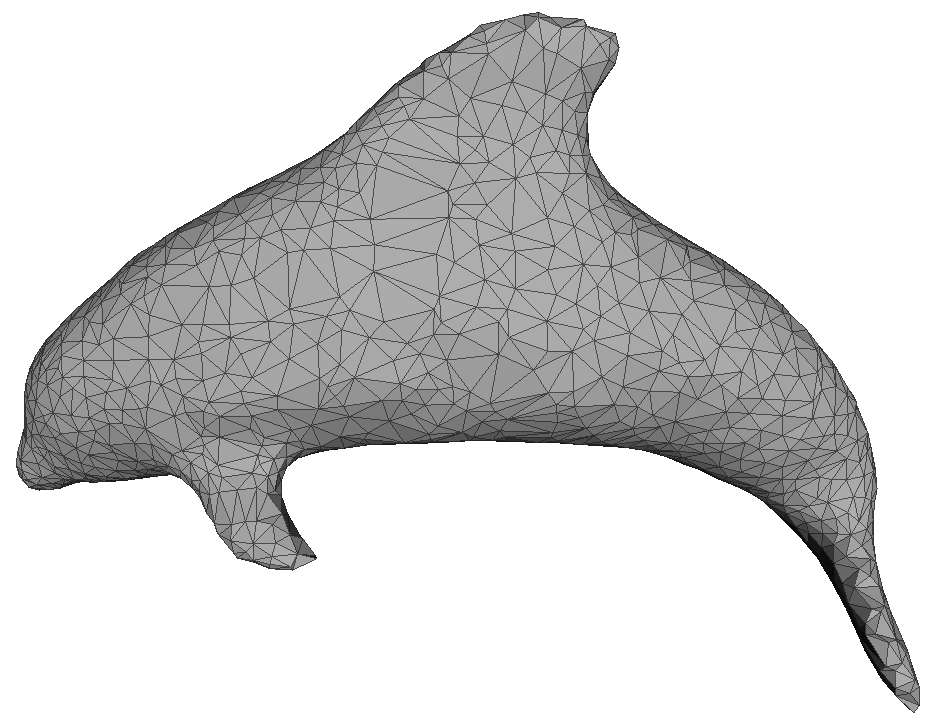}
 \includegraphics[width=0.52\textwidth]{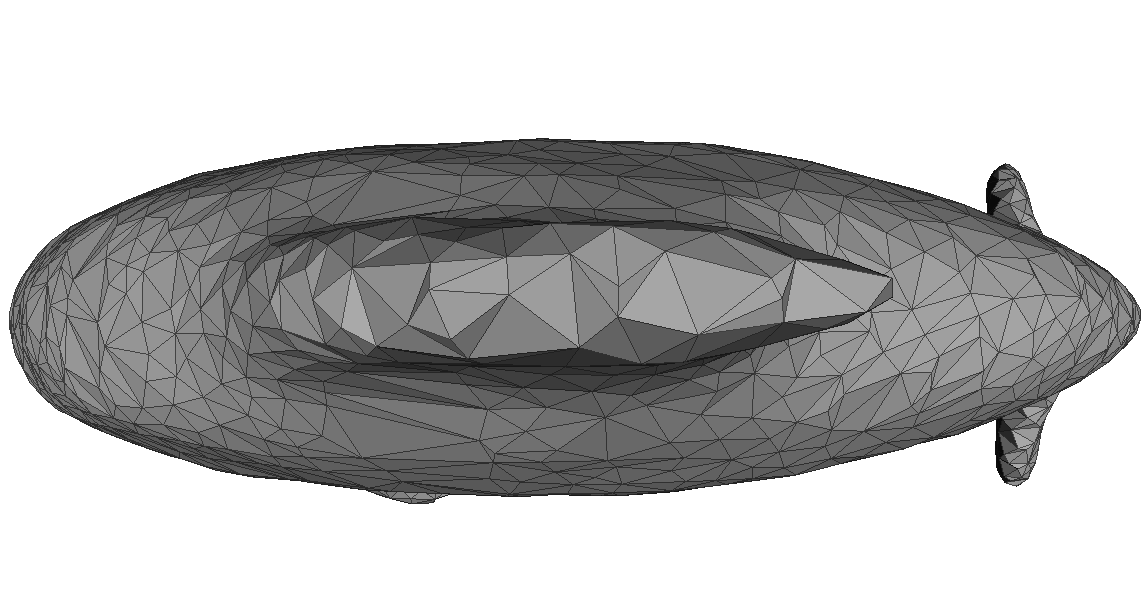} \\ \vspace{6mm}
 \includegraphics[width=0.42\textwidth]{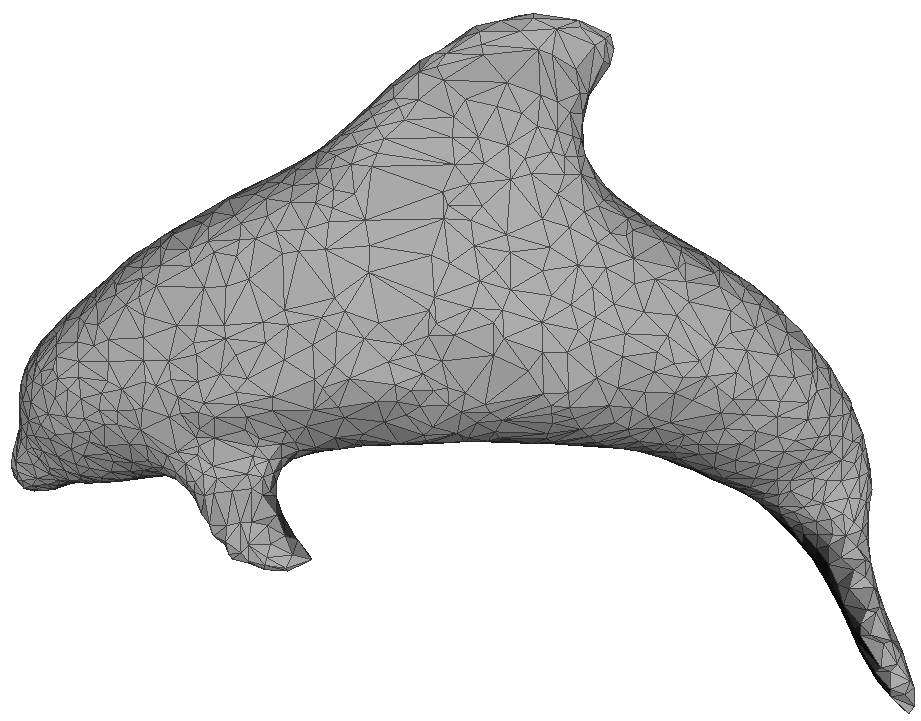}
 \includegraphics[width=0.52\textwidth]{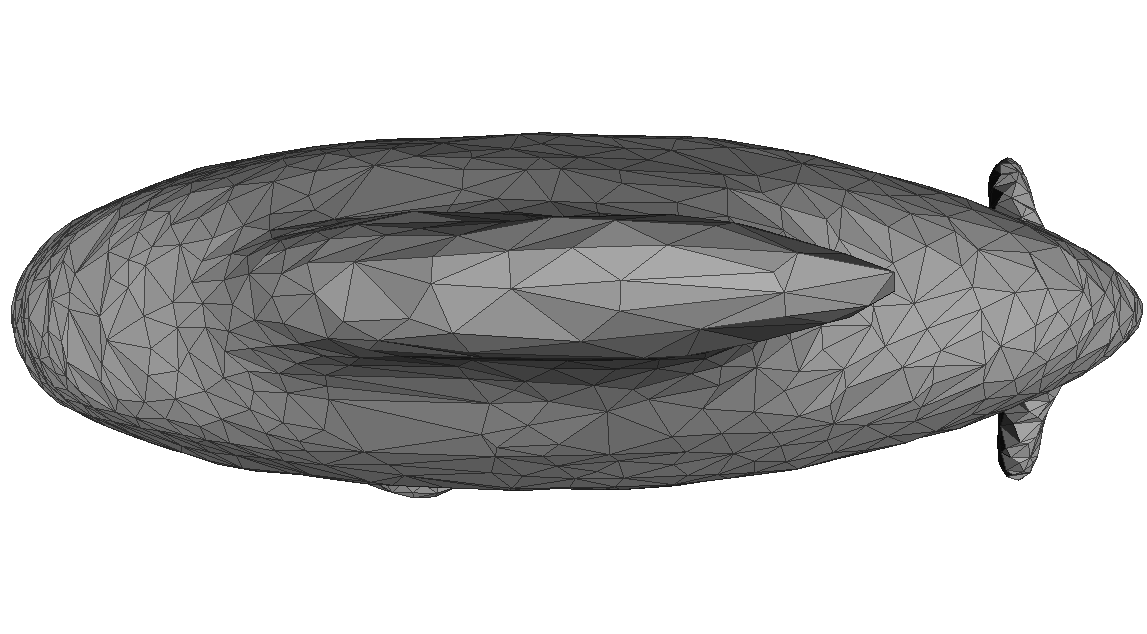}
 \caption{Remeshing the dorsal fin of a dolphin surface. Top: The side view and the top view of the dolphin with a Delaunay triangulation. Bottom: The side view and the top view of the remeshed dolphin obtained by our proposed method. It can be easily observed that the dorsal fin of the dolphin becomes more prominent after the remeshing procedure. }
 \label{fig:dolphin_remeshed}
\end{figure}

Another example of a dolphin surface is shown in Figure \ref{fig:dolphin_remeshed}. The initial triangulation of the dolphin surface is Delaunay. It can be easily observed that the triangles at the dorsal fin of the dolphin are too regular and do not follow the geometry of the dorsal fin. This makes the shape of the dorsal fin non-smooth. To improve the visualization, we define the quasiconformal dilation $K$ as
\begin{equation}
 K(T) = \left\{ \begin{array}{ll}
                  2.5 & \text{ if } T \text{ is at the dorsal fin,}\\
                  1 & \text{ otherwise,}
                \end{array}\right.
\end{equation}
for all triangle elements $T$. Two points $p_1$, $p_2$ are again manually selected at the two ends of the dorsal fin for controlling the direction of the distortion. Then, we apply our proposed remeshing framework with the quasiconformal dilation $K$. This results in a remeshed dolphin surface as shown in Figure \ref{fig:dolphin_remeshed}. It is noteworthy that even without any changes in the positions of the vertices, our adaptive remeshing result significantly enhances the visual quality of the dorsal fin of the dolphin surface.

\begin{figure}[!t]
 \centering 
 \includegraphics[width=0.47\textwidth]{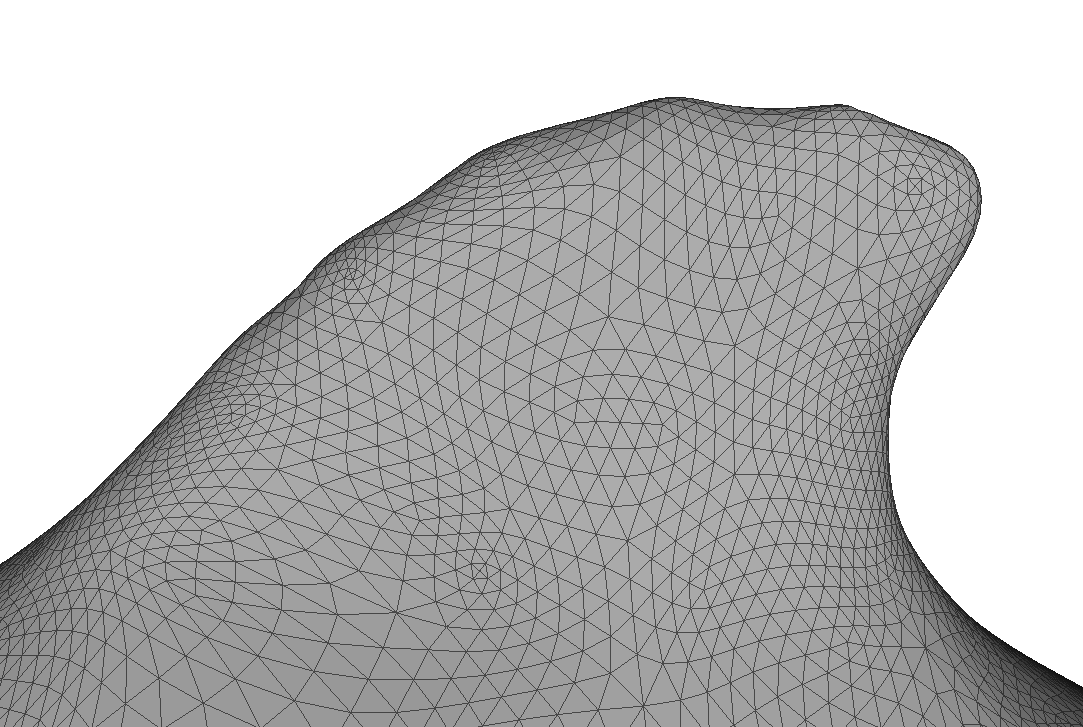}
 \includegraphics[width=0.47\textwidth]{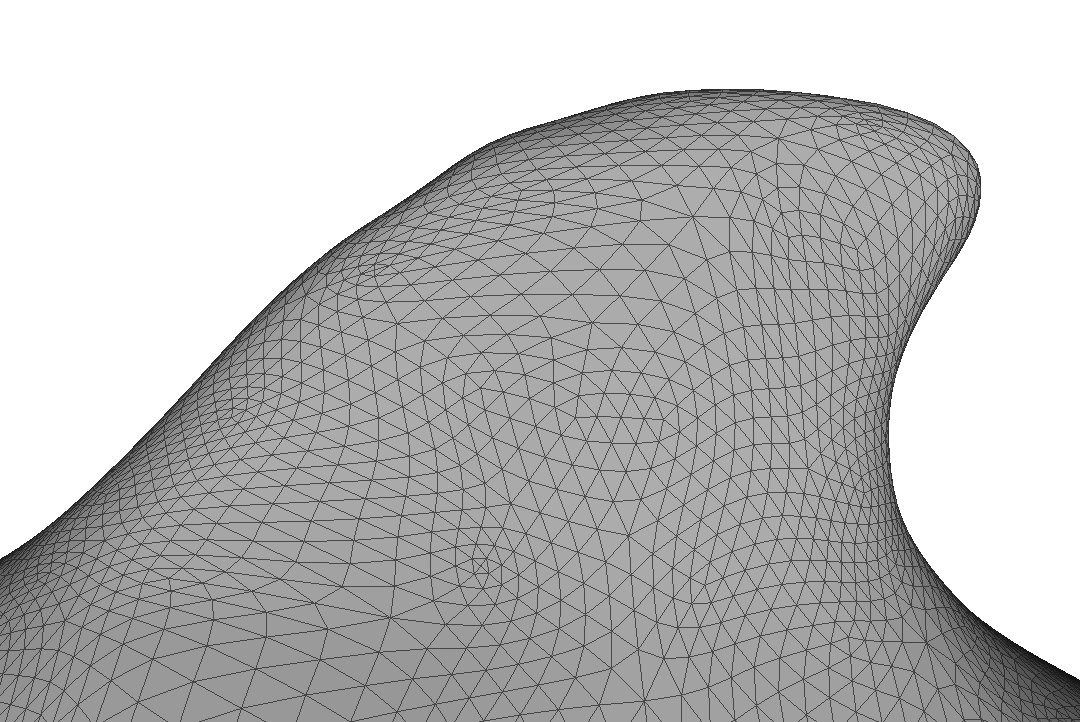}
 \caption{Surface subdivision at the dorsal fins of the two triangulations. Left: The result built upon the original Delaunay triangulation of the dolphin surface. Right: The result built upon the new triangulation obtained by our remeshing framework.}
 \label{fig:subdivision}
\end{figure}

To further highlight the advantage of our remeshing framework, we consider applying a surface subdivision algorithm on the original and the remeshed dolphin surfaces. The LS3 Subdivision Surface Algorithm \cite{Boye10} with Loop's weight \cite{Loop87} is applied. The algorithm is a built-in function in MeshLab. Figure \ref{fig:subdivision} shows the subdivision results. It can be observed that the original Delaunay triangulation of the dolphin surface does not result in a smooth dorsal fin while our triangulation does. This comparison reflects the importance of our adaptive remeshing framework.

\section{Conclusion} \label{conclusion}
In this work, we have developed a fast spherical quasiconformal parameterization algorithm, abbreviated as FSQC, for genus-0 closed surfaces. By appropriately defining the concept of quasiconformal dilation on each triangle element of a mesh, we have proposed a computational scheme for computing a spherical quasiconformal parameterization that satisfies the prescribed quasiconformal distortion. Experimental results have demonstrated the efficiency and accuracy of our algorithm. Furthermore, the FSQC algorithm can be applied for remeshing genus-0 closed surfaces to enhance their visual quality. The effectiveness of our proposed remeshing framework has been illustrated by two remeshing experiments. In the future, we aim to extend the quasiconformal parameterization algorithm for adaptively remeshing high-genus surfaces.

\end{document}